\documentclass[12pt]{article}

\usepackage[makeindex,split,idxcommands]{splitidx}
\usepackage{bbm}
\usepackage{amssymb}

\usepackage{amsfonts}
\usepackage{graphicx}
\usepackage{verbatim}
\usepackage{enumerate}
\usepackage[intlimits]{amsmath}
\usepackage{fullpage}
\usepackage[numbers]{natbib}
\usepackage{amsthm}
\makeindex

\newindex[Notation Index]{symbols}
\newindex[Subject Index]{notions}

\allowdisplaybreaks

\newif\ifhyper\IfFileExists{hyperref.sty}{\hypertrue}{\hyperfalse}
\ifhyper\usepackage{hyperref}\fi

\newif\ifdraft
\drafttrue
\numberwithin{equation}{section}
\numberwithin{figure}{section}
\usepackage[english]{cleveref}

\newtheorem{theorem}{Theorem}
\numberwithin{theorem}{section}
\newtheorem{corollary}[theorem]{Corollary}
\newtheorem{lemma}[theorem]{Lemma}

\newtheorem{definition}[theorem]{Definition}

\theoremstyle{definition}

\newtheorem{remark}[theorem]{Remark}
\newtheorem{claim}[theorem]{Claim}

\newcommand{\E}{\mathbb{E}}

\newcommand{\G}{\mathbb{G}}

\newcommand{\PP} {{\mathbb P}}
\newcommand{\R}{\mathbb{R}}

\newcommand{\N}{\mathbb{N}}

\newcommand{\EEE}{\mathcal{E}}
\newcommand{\WW}{{\bf W}}
\newcommand{\GG}{{\bf G}}

\newcommand{\I}{\mathbb{I}}
\newcommand{\PPP} {{\mathcal P}}

\newcommand{\WWW}{\mathcal{W}}
\newcommand{\QQ} {{\mathcal Q}}
\newcommand{\UU}{{\bf U}}
\newcommand{\VV}{{\bf V}}
\newcommand{\FF}{{\bf F}}
\newcommand{\HHH}{{\bf H}}
\newcommand{\hh}{\mathbb{H}}
\newcommand{\RRR}{\mathcal{R}}
\newcommand{\Ss}{\mathcal{S}}
\newcommand{\TTT}{\mathcal{T}}
\newcommand{\FFF}{\mathbf{F}}

\def\argmax{\mathrm{argmax}}
\def\sgn{\mathrm{sgn}}

\def \_reg {\rightarrow_{\bf reg}}

\def\maxdeg/{\Delta}

\def\tf{{\rm tf}}

\allowdisplaybreaks

\begin{document}
\newcommand{\na}{\mathbf{a}}
\newcommand{\nb}{\mathbf{b}}
\newcommand{\nc}{\mathbf{c}}
\newcommand{\inj}{\mathrm{inj}}
\newcommand{\ind}{\mathrm{ind}}
\newcommand{\Sym}{\mathrm{Sym}}
\newcommand{\Pd}{\mathrm{Pd}}
\newcommand{\du}{\mathrm{d}}

\title{\textbf{Complexity of Nondeterministic Graph Parameter Testing}}
\author{Marek Karpinski\thanks{Dept. of Computer Science and the Hausdroff Center for Mathematics, University of Bonn. Supported in part by DFG grants, the Hausdorff grant EXC59-1/2. E-mail: \textrm{marek@cs.uni-bonn.de}}
\and 
Roland Mark\'o\thanks{Hausdorff Center for Mathematics, University of Bonn. Supported in part by a Hausdorff scholarship. E-mail: \textrm{roland.marko@hcm.uni-bonn.de}}}

\date{}
\maketitle
\begin{abstract}

We study the sample complexity of nondeterministically testable graph parameters and improve existing bounds on it by several orders of magnitude. The technique used would be also of independent interest. We also discuss the special case of weak nondeterministic testing for uniform hypergraphs of arbitrary order.

\end{abstract}

\section{Introduction}

In this paper we investigate the estimation of graph parameters by means of uniform vertex sampling. We focus on the novel concept of nondeterministic graph parameter and property testing. 

We call a non-negative function on the set of labeled simple graphs a \emph{graph parameter} if it is invariant under graph isomorphism, i.e. relabeling of vertices.  
We define parameters of edge-$k$-colored directed graphs, which will be considered in this paper as loop-free and complete in the sense that each directed edge is present and carries exactly one color, and introduce analogously the concept also for graphons, the limit objects of dense graphs \cite{BCL}, \cite{LSzlim}. The central characteristic of parameters investigated in the current paper is whether it is the possible to estimate a given value of the parameter up to a desired accuracy via uniform sampling of bounded size that is independent from the size of the input graph, in short, is the parameter testable. If the answer to this question is positive, then we can ask for the smallest sample size that is sufficient for this purpose. For a graph $G$ (directed and $k$-colored possibly) the expression $\G(q,G)$ denotes the random induced subgraph of $G$ with the vertex set chosen uniformly among all subsets of $V(G)$ that have cardinality $q$.

\begin{definition}
The graph parameter $f$ is testable if for any $\varepsilon>0$ there exists a positive integer $q_0(\varepsilon)$ such that for any $q\geq q_0(\varepsilon)$ and simple graph $G$ with at least $q_0(\varepsilon)$ nodes 
\begin{align*}
\PP(|f(G)-f(\G(q,G))| >\varepsilon) < \varepsilon.
\end{align*}
The smallest function $q_0$ satisfying the previous inequality is called the sample complexity of $f$ and is denoted by $q_f$.
The testability of parameters of $k$-colored directed graphs and uniform hypergraphs is defined analogously.
\end{definition}

From now on colored means edge-colored if not noted otherwise and we generally assume that $q_f(\varepsilon)\geq \max\{1, 1/\varepsilon\}$. An a priori weaker notion than testability is the second cornerstone of the current work, it was introduced in \cite{LV}.

\begin{definition}\label{ch6:defndtest}
	The graph parameter $f$ is non-deterministically testable if there exist integers $k \geq m$ and a testable $k$-colored directed graph parameter $g$ called witness such that for any simple graph $G$ the value $f(G)=\max_{\GG} g(\GG)$ where the maximum goes over the set of $(k,m)$-colorings of $G$.  The edge-$k$-colored directed graph $\GG$ is a $(k,m)$-coloring of $G$, if after erasing all edges of $\GG$ colored with an element of $[m+1, \dots, k]$ and discarding the orientation, coloring, and multiplicity of the remaining edges we end up with $G$. We say in this case that $G$ is the shadow of $\GG$. 
\end{definition}
The corresponding definition for $r$-uniform hypergraphs (in short, $r$-graphs)  is analogous. The choice of maximizing over the $g$-values in \Cref{ch6:defndtest} is somewhat arbitrary, in a more general sense we could have $f(G)=g(\argmax_{\GG} L(g(\GG)))$ for any $\alpha$-H\"older continuous function $L$ from $\R$ to $\R$. Also, stronger formulations of being  a witness can be employed, such as permitting only undirected instances or imposing $k=2m$. 

\subsection{Previous work}
The problem regarding the relationship of the class of parameters that are testable and those who are non-deterministically testable was first studied in the framework of dense graph limits and property testing by Lov\'asz and Vesztergombi \cite{LV} in the spirit of the general ``P vs. NP'' question, that is a central problem in theoretical computer science. Using the particular  notion of nondeterminism above they were able to prove that any non-deterministically testable graph property is also testable, which implies the analogous statement for parameters. 
\begin{theorem}\cite{LV}
	Every non-deterministically testable graph property $\PPP$ is testable. The same holds in parameter testing.
\end{theorem}
However, no explicit relationship  was provided between on one hand, the sample size required for estimating the $f$ value, and on the other, the two factors, the number of colors $k$ and $m$, and the sample complexity of the witness $g$. The reason for the non-efficient characteristic of the result is that the authors exploited various consequences of the next remarkable fact.

Graphons are bounded symmetric measurable functions on the unit square, their cut norm $\|.\|_\square$ given in \Cref{ch3:defnorm} below. At this point we wish to stress that it is weaker than the $L^1$-norm, and the $\delta_\square$-distance induced by it combined with an optimal overlay has a compact unit ball, this is not the case for the $\delta_1$-distance generated analogously by the $L^1$-norm. 

\textbf{Fact.}  If $(W_n)_{n \geq 1}$ is a sequence of graphons and $\|W_n\|_\square \to 0$ when $n$ tends to infinity, then for any measurable function $Z\colon [0,1]^2 \to [-1,1]$ it is true that $\|W_nZ\|_\square \to 0$, where the product is taken point-wise.

Although the above statement is true for all $Z$, the convergence is not uniform and its rate depends heavily on the structure of $Z$.

The relationship of the magnitude of the sample complexity of a nondeterministically testable property $\PPP$ and its colored witness $\QQ$  was analyzed by \citet{GS} relying on Szemer\'edi's Regularity Lemma and its connections to graph property testing unveiled by \citet*{AFNS}. In the upper bound given in \cite{GS} the height of the exponential tower was not bounded and growing as a function of the inverse of the accuracy, $1/\varepsilon$, the main result of \cite{GS} for parameters can be rephrased as follows.
\begin{theorem}\cite{GS}
	Every non-deterministically testable graph parameter $f$ is testable. If the sample complexity of the witness parameter $g$ for each $\varepsilon >0$ is $q_g(\varepsilon)$, then the sample complexity of $f$ for each $\varepsilon>0$ is at most $\tf(c q_g(\varepsilon/2))$ for some universal constant $c>0$, where $\tf(t)$ is the exponential tower of twos of height $t$.
\end{theorem}

\subsection{Our contribution}
In the current paper
we improve on the result of \cite{GS} by using a weaker type of regularity approach which eliminates the tower-type dependence on the sample complexity of the witness parameter. The function $\exp^{(t)}$ stands for the $t$-fold iteration of the exponential function ($\exp^{(0)}=\mathrm{id}$). Our main result is the following. 
\begin{theorem}\label{ch6:graphmain}
	Let $f$ be a nondeterministically testable simple graph parameter with witness parameter $g$ of $k$-colored digraphs, and let the corresponding sample complexity be $q_g$. Then $f$ is testable with sample complexity $q_f$, and  there exists a constant $c>0$ only depending on $k$ but not on $f$ or $g$ such that for any $\varepsilon>0$ the inequality $q_f(\varepsilon) \leq \exp^{(3)}(cq^2_g(\varepsilon/2))$ holds.
\end{theorem}

We also investigate the case where only node colored graphs and their parameters serve as nondeterministic certificate, and improve on the above upper bound in that setting. Moreover, we extend the method to be able to deal with hypergraphs of higher rank, see below for the precise formulation. 

\subsection{Outline of the paper}

This paper is organized as follows. In Section \ref{sec.not} we introduce the basic notation related to dense graph limit theory that is necessary to conduct the proof of the main result in Theorem \ref{ch6:graphmain}, and we will also state and prove the main ingredient of the proof, our intermediate regularity lemma, that might be of interest on its own right. Section \ref{sec.thm} continues with the proof of Theorem \ref{ch6:graphmain}, while in Section \ref{ch6:sec:weaknd} we treat a special case of the non-deterministic testing notion applied in the current paper.

\section{Graph limits and regularity lemmas}\label{sec.not}

First we provide the definition of graph convergence via subgraph densities. For the simple graphs $F$ and $G$ let $\hom(F,G)$ denote the number of maps $\phi \colon V(F) \to V(G)$ that preserve binary relationships, that is for each $u,v \in V(F)$ we have $\phi(u)\phi(v) \in E(G)$ if and only if when $uv \in E(F)$, in particular, a homomorphism has to be injective. In some previous works only the presence of edges had to be sustained by $\phi$. Furthermore, let $t(F,G)=\frac{\hom(F,G)}{|V(G)|^{|V(F)|}}$ denote the subgraph density of $F$ in $G$. The density $t(\FF,\GG)$ in the case of $k$-colored digraphs is defined analogously. 

\begin{definition}\cite{LSzlim}\label{defconv}
Let $(G_n)_{n \geq 1}$ be a sequence of simple graphs. It is said to be convergent if for every simple graph $F$ the numerical sequences $(t(F,G_n))_{n \geq 1}$ converge to some limit. Convergence is defined in the case of sequences of $k$-colored digraphs analogously.
\end{definition} 

We now describe the space of limit objects of simple graphs in the sense of \citet{LSzlim}. Let $\mathcal W_0$ be the set of all bounded measurable functions $W \colon [0,1] \times [0,1] \to \R$, these objects are called \emph{kernels}. The subspace $\mathcal W$  of $[0,1]$-valued elements of $\mathcal W_0$ that are symmetric in the sense that $W(x,y)=W(y,x)$ for all $x,y \in [0,1]$ is the space of \emph{graphons}. 
The space of $k$-colored directed graphons can be described in a similar, though more involved way. Let $\mathcal W_0^{(k)}$ be the set of $k \times k$-tuples $\WW=(W^{(\alpha,\beta)})_{\alpha,\beta \in [k]}$ of kernels referred to as \emph{$k$-colored dikernels}. The subspace $\mathcal W^{(k)}$ of $\mathcal W_0^{(k)}$ whose components  obey a symmetry in the sense that $W^{(\alpha,\beta)}(x,y)=W^{(\beta,\alpha)}(y,x)$ for each $x,y \in [0,1]$, are non-negative, and additionally satisfy $\sum_{\alpha,\beta \in [k]} W^{(\alpha,\beta)} (x,y)=1$ for each $x,y \in [0,1]$ is referred to as the space of \emph{$k$-colored directed graphons}. Note that for each set $[m]$ with $m \leq k$ the function $W'(x,y)=\sum_{\alpha,\beta \in [m]} W^{(\alpha,\beta)} (x,y)$ is a graphon when $\WW$ is $k$-colored directed graphon, furthermore, each graphon $W$ can be regarded as a $2$-colored directed graphon by setting $W^{(1,2)}=W^{(2,1)}=0$ and $W^{(1,1)}=W$ everywhere.

A step function is a real-valued function on $[0,1]^2$ that is constant on the product sets $P_i \times Q_j$ for some pair of partitions $\PPP$ and $\QQ$ of $[0,1]$ into the same number of classes  referred to as steps, the step function is proper when $\PPP=\QQ$.  For a partition $\PPP$ the integer $t_\PPP$\sindex[symbols]{t@$t_\PPP$} denotes the number of its classes. We call a partition $\PPP$ of $[0,1]$ a \emph{canonical $n$-partition} if its classes are the intervals $P_i=[\frac{i-1}{n},\frac{i}{n})$ for each $i \in [n]$.  If the canonical $n$-partition refines a partition, then we speak of an $\mathcal{I}_n$-partition, and an $\mathcal{I}_n$-set is the union of some classes of the canonical $n$-partition. Further, a measure-preserving map from $[0,1]$ to $[0,1]$ is referred to as an $\mathcal{I}_n$-permutation if it corresponds to a permutation of the classes of the canonical $n$-partition. Functions on $[0,1]$ and $[0,1]^2$ are called $\mathcal{I}_n$-functions if they are constant on the classes of the canonical $n$-partition and products of those, respectively. We require these concepts to be able to relate graphs on different vertex sets to each other in a simple and computationally efficient way.

We can associate to each simple graph $G$ on $n$ vertices a graphon $W_G$ that is a step function with the steps forming the canonical $n$-partition and taking the value $1$ on $P_i \times P_j$ whenever $ij \in E(G)$ and $0$ otherwise. Similarly, for a $k$-colored directed $\GG$ we can define $\WW_\GG$ as the step function with the same steps as above and set $W_{\GG}^{(\alpha, \beta)}$ to $1$ on $P_i \times P_j$ when $(i,j)$ is colored by $\alpha$ and $(j,i)$ by $\beta$ in $\GG$, and to $0$ otherwise. Here on the diagonal cubes $P_i \times P_i$ we require the special color $\iota$ (standing for undefined color), the overall measure of the diagonal cubes is $O(1/|V(\GG)|)$.

Next we define the sampling process for the objects in consideration.

\begin{definition}\label{defsamp}
Let $q\geq 1$, $G$ be a simple graph and $S$ be a random subset of $V(G)$ chosen among all subsets of cardinality $q$ uniformly. Then $\G(q,G)$ denotes the random induced subgraph of $G$ on $S$. For a $k$-colored directed graph $\GG$ the random subgraph $\G(q, \GG)$ is defined analogously, same applies for $r$-uniform hypergraphs for arbitrary $r$. 

Let $W$ be a graphon and $q\geq 1$,  furthermore,  $(X_i)_{i \in [q]}$ and $(Y_{ij})_{ij \in {[q] \choose 2}}$ be mutually pair-wise independent uniform $[0,1]$ random variables. Then the random graph $\G(q,W)$ has vertex set $[q]$ and an edge runs between the vertices $i$ and $j$ if $Y_{ij}\geq W(X_i, X_j)$. The random $k$-colored directed graph $\G(q, \WW)$ has also vertex set $[q]$, and further, conditioned on the choice of $(X_i)_{i \in [q]}$, the colors for the edges in the two directions are chosen independently for all pairs $ij \in {[q] \choose 2}$ of vertices, the event that $(i,j)$ carries the color $\alpha$ and at the same time $(j,i)$ carries the color $\beta$ has probability $W^{(\alpha, \beta)}(X_i, X_j)$. \end{definition}

Note that in $\G(q,\WW)$ the colors of $(i,j)$ and $(j,i)$ are not even conditionally independent as random objects.

The density of a simple graph $F$ with vertex set $[q]$ in a graphon $W$ is defined as 
\begin{align*}
t(F,W)=\int_{[0,1]^q} \prod_{ij \in E(F)} W(x_i, x_j) \prod_{ij \notin E(F)} (1-W(x_i, x_j)) \du x,
\end{align*}
and the density of a colored digraph $\FF$ with the same  vertex set as above described as a matrix with color entries in $\WW$ is given as
\begin{align*}
t(\FF,\WW)=\int_{[0,1]^q} \prod_{\substack{ij: \FF(i,j)=\alpha \\
\FF(j,i)=\beta }} W^{(\alpha, \beta)}(x_i, x_j) \du x,
\end{align*}
where the product in the integral is taken over all unordered pairs $ij \in {[k] \choose 2}$.

The next theorem, first proven in \cite{LSzlim}, states that the graphons truly represent the limit space of simple graphs. For the proof of the more general cases, see \cite{DJ}, \cite{KM1}, and \cite{LSzcom}. 

\begin{theorem}\label{thmconv} \cite{LSzlim}, \cite{LSzcom}
If $(G_n)_{n \geq 1}$ is a convergent sequence of simple graphs, then there exists a graphon $W$ such that for every simple graph $F$ we have $t(F,G_n) \to t(F,W)$, when $n$ tends to infinity. Similarly, if $(\GG_n)_{n \geq 1}$ is a convergent sequence of $k$-colored directed graphs, then there exists a $k$-colored digraphon $\WW$ such that for every $k$-colored digraph $\FF$ it holds that $t(\FF,\GG_n) \to t(\FF,\WW)$.
\end{theorem}

We proceed by enumerating the norms and distances that are relevant for the current work and are related to the graph limit theory and parameter testing. 

\begin{definition}\label{ch3:defnorm}
	The cut norm\sindex[notions]{cut norm!matrix} of a real $n \times n$ matrix $A$ is 
	\begin{align*}
	\|A\|_\square=\frac{1}{n^2} \max_{S,T \subset [n]} \left| A(S,T) \right|,
	\end{align*}\sindex[symbols]{a@$\Vert A \Vert_\square$,$\Vert W \Vert_\square$}
	where $A(S,T)=\sum_{s \in S, t \in T} A(s,t)$.\sindex[symbols]{a@$A(S,T)$}
	
	The cut distance\sindex[notions]{cut distance} of two labeled simple graphs $F$ and $G$ on the same vertex set $[n]$ is 
	\begin{align*}
	d_\square(F,G)= \|A_F-A_G\|_\square,
	\end{align*}\sindex[symbols]{d@$d_\square(F,G)$}
	where $A_F$ and $A_G$ stand for the respective adjacency matrices.
	
	The cut norm of a kernel $W$ is 
	\begin{align}\label{ch3:cutdefeq}
	\|W\|_\square= \max_{S,T \subset [0,1]} \left|\int_{S \times T} W(x,y) \du x \du y \right|,
	\end{align}
	where maximum is taken over all pairs of measurable sets $S$ and $T$. We speak of the $n$-cut norm of kernels\sindex[notions]{cut norm!$n$-cut norm} when the maximum in (\ref{ch3:cutdefeq}) is only taken over pairs of $\mathcal{I}_n$-sets, it is denoted by $\|W\|^{\langle n \rangle}_\square$\sindex[symbols]{w@$\Vert W\Vert^{\langle n \rangle}_\square$}. The cut norm of a $k\times k$-tuple of kernels $\WW=(W^{(\alpha, \beta)})_{\alpha, \beta=1}^k$ is 
	\begin{align*}
	\|\WW\|_\square=\sum_{\alpha, \beta=1}^{k} \|W^{(\alpha, \beta)}\|_\square.
	\end{align*}
	
	The cut distance of two graphons $W$ and $U$ is
	\begin{align*}
	\delta_\square(W, U) = \inf_{\phi, \psi} \|W^\phi-U^\psi\|_\square,
	\end{align*}\sindex[symbols]{d@$\delta_\square(U, W)$,$\delta_\square(G,H)$}
	where the infimum runs over all pairs of measure-preserving map from $[0,1]$ to $[0,1]$, and the graphon $W^\phi$ is defined as $W^\phi(x,y)= W(\phi(x),\phi(y))$.
	Similarly, for $k$-colored directed graphons $\WW$ and $\UU$ we have
	\begin{align*}
	\delta_\square(\WW, \UU) = \inf_{\phi, \psi} \|\WW^\phi-\UU^\psi\|_\square,
	\end{align*}
	with the difference being component-wise.
	The cut distance for arbitrary unlabeled graphs\sindex[notions]{cut distance!graph} $F$ and $G$ is
	\begin{align*}
	\delta_\square(F,G)=\delta_\square(W_F,W_G),
	\end{align*}
	the definitions for the colored directed version is identical.
	Another variant is for the case when $V(F)=[m]$ and $V(G)=[n]$ such that $m$ is a divisor of $n$. Then 
	\begin{align*}
	\hat \delta_\square^{\langle n \rangle} (F,G)=\min_\phi d_\square (F[n/m],G^\phi),
	\end{align*}
	where $F[t]$ is the $t$-fold equitable blow up of $F$ and the minimum goes over all node relabellings $\phi$ of $G$. In the case $n=m$ we omit the upper index and use $\hat \delta_\square$.
\end{definition} 

In fact, $\hat \delta_\square$ and $\delta_\square$ define only pseudometrics, graphs have distance zero whenever they have equitable blow-ups that are isomorphic.  For graphons we introduce the term graphon equivalence\sindex[notions]{graphon!equivalence} for the case whenever the $\delta_\square$ distance is $0$, but will refer to the above with a slight abuse of notation as proper distances.

Observe that for two graphs $F$ and $G$ on the common node set $[n]$ the distance $d_\square(F,G)=\|W_F-W_G\|_\square=\|W_F-W_G\|^{\langle n \rangle}_\square$. Also note that in general for $F$ and $G$ with identical vertex cardinalities $\delta_\square(F,G)$ is not necessarily equal to $\delta^{\langle n \rangle}_\square(F,G)$, however in \cite{BCL} it was demonstrated that $ \delta_\square(F,G)\leq \hat \delta_\square(F,G)\leq 32(\delta_\square(F,G))^{1/67}$.

An important property of the distances introduced above is that subgraph densities are uniformly continuous in the topology defined by them.

\begin{lemma}\cite{BCL}\sindex[notions]{Counting Lemma}
	\label{ch3:countlemma}
	Let $U$ and $W$ be two graphons. Then for every simple graph $F$ on $q$ vertices we have
	\begin{align*}
	|t(F,W)-t(F,U)| \leq {q \choose 2} \delta_{\square}(U,W). 
	\end{align*}
	The analogous result holds for $k$-colored digraphons.
\end{lemma}

The connection to graph limits is given in the next theorem from \cite{BCL}.

\begin{theorem} \label{ch3:thmdist} \cite{BCL}
	A graph sequence $(G_n)_{n \geq 1}$ (a $k$-colored directed graph sequence $(\GG_n)_{n \geq 1}$, respectively) is convergent if and only if it is Cauchy in the $\delta_\square$ metric.  
\end{theorem}

%
A remarkable feature of the $\delta_\square$ distance is that the deviation of a sampled graph from the original graph or graphon can be upper bounded by a function that decreases logarithmically in the inverse of the sample size. Originally, this result was established to verify \Cref{ch3:thmdist}.

\begin{lemma}\label{ch3:samp}\cite{BCL}
	Let $\varepsilon>0$ and let $U$ be a graphon. Then for $q \geq 2^{100/\varepsilon^2}$ we have
	\begin{align}
	\PP\left(\delta_\square(U, \G(q,U)) \geq \varepsilon\right) \leq \exp\left(-4^{100/\varepsilon^2} \frac{\varepsilon^2}{50}\right). 
	\end{align}
	
\end{lemma}

We turn our attention to the continuous formulation of the Regularity Lemma of Frieze and Kannan \cite{FK} in the graphon space. For a partition $\PPP$ of $[0,1]$ and a kernel $W$  we obtain $W_\PPP$\sindex[symbols]{w@$W_\PPP$} from $W$ by averaging on every rectangle given by product sets from $\PPP$.

\begin{lemma}[Weak Regularity Lemma for kernels]\label{ch3:wreg} \cite{FK}, \cite{LSzreg} 
	For every $\varepsilon>0$ and $W \in \mathcal W_0$ there exists a partition $\PPP=(P_1, \dots, P_m)$ of $[0,1]$ into $m \leq 2^{\frac{8}{\varepsilon^2}}$ parts, such that 
	\begin{align}
	\| W-W_{\PPP} \|_\square \leq \varepsilon \|W\|_2.
	\end{align}
\end{lemma}\sindex[notions]{Weak Regularity Lemma!$2$-kernel}

By slight adaptation of the original proof in  \cite{FK} of the above result we obtain the version for $k$-colored digraphons. 

\begin{lemma}[Weak Regularity Lemma for $k$-colored directed graphons]\label{ch3:wregk}
	For every $\varepsilon>0$  and $k$-colored digraphon $\WW$ there exists a partition $\PPP=(P_1, \dots, P_m)$ of $[0,1]$ into $m \leq 2^{k^4\frac{8}{\varepsilon^2}}=t'_k(\varepsilon)$  parts, such that 
	\begin{align}
	d_\square(\WW,\WW_\PPP)=\sum_{\alpha, \beta=1}^k \| W^{(\alpha, \beta)}-(W^{(\alpha, \beta)})_{\PPP} \|_\square \leq \varepsilon. 
	\end{align}
	When $\WW=\WW_\GG$ for a $k$-colored digraph $\GG$ with vertex cardinality $n$, then one can require in the above statement that $\PPP$ is an $\mathcal{I}_n$-partition.
\end{lemma}\sindex[notions]{Weak Regularity Lemma!$(k,2)$-digraphon}


The following kernel norm shares some useful properties with the cut-norm. Most prominently it admits a regularity lemma that outputs a partition whose number of classes is considerably below the tower-type magnitude in the desired accuracy. On the other hand, it does not admit a straight-forward definition of a related distance by calculating the norm of the difference of two optimally overlayed objects as in  \Cref{ch3:defnorm}. This is the result of the general assumption that the partition $\PPP$ below involved in the definition always belongs to one of the graphons whose deviation we wish to estimate. Therefore a relabeling of this graphon should also act on $\PPP$, hence symmetry fails. Its advantages in comparison to the cut norm will become clearer in the proof of the main result below. 
\begin{definition}\label{ch3:cutpnorm}
	Let $W$ be a kernel and $\PPP=(P_1, \dots, P_t)$ a partition of $[0,1]$. Then the cut-$\PPP$-norm\sindex[notions]{cut-$\PPP$-norm} of $W$ is 
	\begin{align}
	\|W\|_{\square\PPP}= \max_{S_i, T_i \subset P_i} \sum_{i,j=1}^t \left|\int_{S_i \times T_j} W(x,y) \du x \du y \right|.
	\end{align}\sindex[symbols]{w@$\Vert W\Vert_{\square\PPP}$,$\Vert H\Vert_{\square\PPP}$}
	For two kernels $U$ and $W$ let $d_{W,\PPP}(U)$\sindex[symbols]{d@$d_{W,\PPP}(U)$} denote the cut-$\PPP$-deviation\sindex[notions]{cut-$\PPP$-deviation} of $U$ with respect to $W$ that is defined by
	\begin{align}
	d_{W,\PPP}(U)= \inf_{\phi} \|U^\phi-W\|_{\square\PPP},
	\end{align}
	where the infimum runs over all measure preserving maps from $[0,1]$ to $[0,1]$.
	
	For $n\geq 1$, a partition $\PPP$ of $[n]$ and a directed weighted graph $H$ the cut-$\PPP$-norm of $H$ on $[n]$ is defined as   
	\begin{align}
	\|H\|_{\square\PPP}=\|W_H\|_{\square\PPP'},
	\end{align}
	where $\PPP'$ is the partition of $[0,1]$ induced by $\PPP$ and the map $j \mapsto [\frac{j-1}{n},\frac{j}{n})$.
\end{definition}

The definition for the $k$-colored version is analogous.

\begin{definition}\label{ch3:cutpnorm2}
	Let $\WW=(W^{(1,1)}, \dots, W^{(k,k)})$ be a $k\times k$ tuple of kernels and $\PPP=(P_1, \dots, P_t)$ a partition of $[0,1]$. Then the cut-$\PPP$-norm\sindex[notions]{cut-$\PPP$-norm!$(k,2)$-dikernel} of $\WW$ is 
	\begin{align}
	\|\WW\|_{\square\PPP}= \sum_{\alpha, \beta=1}^k \|W^{(\alpha, \beta)}\|_{\square\PPP}.
	\end{align}
	For two $k$-colored directed graphons $\UU$ and $\WW$ let $d_{\WW,\PPP}(\UU)$ denote the cut-$\PPP$-deviation\sindex[notions]{cut-$\PPP$-deviation!$(k,2)$-dikernel} of $\UU$ with respect to $\WW$ that is defined by
	\begin{align}
	d_{\WW,\PPP}(\UU)= \inf_{\phi} \|\UU^\phi-\WW\|_{\square\PPP}= \inf_{\phi} \sum_{\alpha, \beta=1}^k \|(U^{(\alpha, \beta)})^\phi-W^{(\alpha, \beta)}\|_{\square\PPP},
	\end{align}
	where the infimum runs over all measure preserving maps from $[0,1]$ to $[0,1]$. 
\end{definition}

It is not hard to check that the cut-$\PPP$-norm is in fact a norm on the space where we identify two kernels when they differ only on a set of measure $0$. From the definition it follows directly that for arbitrary kernels $U$ and $W$, and any partition $\PPP$ we have $\|W\|_\square \leq \|W\|_{\square\PPP} \leq \|W\|_1$ and $\delta_{\square}(U,W) \leq d_{W,\PPP}(U) \leq \delta_{1}(U,W)$, the same is true for the $k$-colored directed version.

\begin{remark} \label{ch3:rem1}
	We present a different description of the cut-$\PPP$-norm of $W$ and $\WW$ respectively that will allow us to rely on results concerning the cut-norm of \Cref{ch3:defnorm} more directly. For a partition $\PPP$ with $t$ classes and $A=(A_{j,l})_{j,l=1}^t \in \{-1,+1\}^{t \times t}$, let $W^A(x,y)=A_{j,l}W(x,y)$ ($\WW^A$ is given by $(W^{(\alpha,\beta)})^{A}(x,y)=A_{j,l} W^{(\alpha,\beta)}(x,y)$ respectively) for $x \in P_j$ and $y \in P_l$. Then $\|W\|_{\square\PPP}=\max\limits_A \|W^A\|_\square$ and $\|\WW\|_{\square\PPP}=\max\limits_A \|\WW^A\|_\square$.
\end{remark}


This newly introduced norm admits a uniform approximation in the following sense that is essential to conduct the proof of \Cref{ch6:graphmain}. 

\begin{lemma}
	\label{ch3:sregk}
	
	For every $\varepsilon>0$, $m_0\colon[0,1] \to \N$, $k \geq 1$ and $k$-colored directed graphon $\WW=(W^{(\alpha, \beta)})_{\alpha, \beta \in [k]}$ there exists a partition $\PPP=(P_1, \dots, P_t)$ of $[0,1]$ into $t \leq\frac{(16m_0(\varepsilon))^{2^{\frac{k^4}{\varepsilon^2}}}}{4}=t_k(\varepsilon, m_0(\varepsilon))$ parts, such that for any partition $\QQ$ of $[0,1]$ into at most $\max\{t, m_0(\varepsilon)\}$ classes we have
	\begin{align}
	\|\WW-\WW_\PPP\|_{\square\QQ} \leq \varepsilon. 
	\end{align}
	If $\WW=\WW_\GG$ for some $k$-colored $\GG$ with $|V(\GG)|=n$, then one can require that $\PPP$ is an $\mathcal{I}_n$-partition. 
	If we want the parts to have equal measure (almost equal in the graph case), then the upper bound on the number of classes is modified to $\frac{((3k)^{12}m_0(\varepsilon)/\varepsilon^4)^{2^{2k^4/\varepsilon^2}}}{(3k)^6/\varepsilon^2}$.
\end{lemma}

\begin{proof}
	Fix an arbitrary function $m_0$, a $\varepsilon>0$, and a $\WW \in \WWW^{(k)}$. We construct a sequence of partitions $\RRR_0, \RRR_1, \dots, \RRR_m$ such that $\RRR_0=[0,1]$ and each $\RRR_{i+1}$ refines the preceding $\RRR_i$. The integer $m$ is a priori undefined.
	
	The construction is sequential in the sense that we assume that we have already constructed $\RRR_{0}, \RRR_1, \dots, \RRR_{i-1}$ before considering the $i$th step of the construction. 
	
	If for $i \geq 1$ there exists a partition $\QQ=(Q_1, \dots, Q_{t_{\QQ}})$ of $[0,1]$ into at most $\max\{t_{\RRR_{i-1}}, m_0(\varepsilon)\}$ classes such that 
	\begin{align}\label{ch3:eq00}
	\|\WW-\WW_{\RRR_{i-1}}\|_{\square\QQ} > \varepsilon, 
	\end{align}
	then we proceed to the construction of $\RRR_i$. In the case of $i=1$ we choose $\QQ$ to have exactly $m_0$ parts of positive measure, this can be achieved since for any refinement $\QQ'$ of $\QQ$ we have $\|\WW-\WW_{\RRR_{i-1}}\|_{\square\QQ'} \geq \|\WW-\WW_{\RRR_{i-1}}\|_{\square\QQ}$.
	The inequality (\ref{ch3:eq00}) implies that there are $\alpha_0, \beta_0 \in[k]$ and measurable subsets $S$ and $T$ of $[0,1]$ such that
	\begin{align}\label{ch3:eq01}
	\sum_{i,j=1}^{t_\QQ} \left|\int_{(S \cap Q_i) \times (T \cap Q_j)} W^{(\alpha_0, \beta_0)}(x,y) - W^{(\alpha_0, \beta_0)}_{\RRR_{i-1}}(x,y) \du x \du y \right| > \varepsilon/k^2.
	\end{align}
	In this case we define $\RRR_{i}$ to be the coarsest common refinement of $\RRR_{i-1}$, $\QQ$, and $\{S,T\}$ for some arbitrary choice of the latter partition and sets satisfying (\ref{ch3:eq01}).
	
	Set $S_j=S \cap Q_j $ and $T_j=T \cap Q_j$ for $j \in [t_\QQ]$ and $U=W^{(\alpha_0, \beta_0)} - W^{(\alpha_0, \beta_0)}_{\RRR_{i-1}}$, further, define the step function	$V=\sum_{j,l \in [t_\QQ]} \sgn(\int_{S_j \times T_l} U ) \I_{S_j \times T_l}$
	
	In this case 
	\begin{align}
	\|\WW_{\RRR_{i}}\|_2^2-\|\WW_{\RRR_{i-1}}\|_2^2 &=\sum_{\alpha, \beta=1}^{k} \langle W^{(\alpha, \beta)}_{\RRR_{i}}, W^{(\alpha, \beta)}_{\RRR_{i}}\rangle - \langle W^{(\alpha, \beta)}_{\RRR_{i-1}}, W^{(\alpha, \beta)}_{\RRR_{i-1}}\rangle \nonumber \\
	& = \sum_{\alpha, \beta=1}^{k} \langle W^{(\alpha, \beta)}_{\RRR_{i}}-W^{(\alpha, \beta)}_{\RRR_{i-1}}, W^{(\alpha, \beta)}_{\RRR_{i}}-W^{(\alpha, \beta)}_{\RRR_{i-1}}\rangle \label{ch3:eq02} \\
	&= 	\sum_{\alpha, \beta=1}^{k} \|W^{(\alpha, \beta)}_{\RRR_{i}}-W^{(\alpha, \beta)}_{\RRR_{i-1}}\|_2^2 \nonumber \\
	&\geq 	 \|W^{(\alpha_0, \beta_0)}_{\RRR_{i}}-W^{(\alpha_0, \beta_0)}_{\RRR_{i-1}}\|_2^2 \nonumber \\
	&\geq \frac{1}{\|V\|_2^2} |\langle W^{(\alpha_0, \beta_0)}_{\RRR_{i}}-W^{(\alpha_0, \beta_0)}_{\RRR_{i-1}}, V  \rangle|^2 \label{ch3:eq04}\\
	&= \frac{1}{\|V\|_2^2} |\langle W^{(\alpha_0, \beta_0)}-W^{(\alpha_0, \beta_0)}_{\RRR_{i-1}}, V  \rangle |^2  \label{ch3:eq05}\\
	&\geq |\langle W^{(\alpha_0, \beta_0)}-W^{(\alpha_0, \beta_0)}_{\RRR_{i-1}}, V  \rangle |^2 \label{ch3:eq06}\\
	&> \varepsilon^2/k^4. \label{ch3:eq07}
	\end{align}
	Here we used first  in (\ref{ch3:eq02}) that $\langle W^{(\alpha, \beta)}_{\RRR_{i-1}}, W^{(\alpha, \beta)}_{\RRR_{i}}\rangle=\langle W^{(\alpha, \beta)}_{\RRR_{i-1}}, W^{(\alpha, \beta)}_{\RRR_{i-1}}\rangle$, since $W^{(\alpha, \beta)}_{\RRR_{i-1}}$ is constant on  $\RRR_{i-1}$ rectangles, and the integral of the two functions is equal on these rectangles. In (\ref{ch3:eq04}) we used the Cauchy-Schwarz inequality, then in (\ref{ch3:eq05}) the fact that $\langle W^{(\alpha_0, \beta_0)}_{\RRR_{i}}, V  \rangle=\langle W^{(\alpha_0, \beta_0)}, V  \rangle$, that is true by $V$ being constant on $\RRR_i$ rectangles and $W^{(\alpha_0, \beta_0)}_{\RRR_{i}}$ and $W^{(\alpha_0, \beta_0)}$ having the same integral value on taken on $\RRR_i$ rectangles. We concluded the calculation in (\ref{ch3:eq06}) by $\|V\|_2 \leq 1$ and in (\ref{ch3:eq07}) using the condition (\ref{ch3:eq01}). 
	
	If for 
	some $i_0\geq 0$ we have 
	\begin{align}
	\|\WW-\WW_{\RRR_{i_0}}\|_{\square\QQ} \leq \varepsilon 
	\end{align}
	for every  partition $\QQ$ of $[0,1]$ into at most $\max\{t_{\RRR_{i_0}}, m_0(\varepsilon)\}$ classes, then we stop the process and set $\PPP=\RRR_{i_0}$  and $m=i_0$.
	
	We have $\|\WW_{\RRR_{j}}\|_2^2\leq \|\WW\|_2^2\leq \|\WW\|_1^2\leq 1$ for each $j\geq 0$, and at each non-terminating step we showed $\|\WW_{\RRR_{i}}\|_2^2-\|\WW_{\RRR_{i-1}}\|_2^2 > \varepsilon^2/k^4$. Therefore by 
	\begin{align*}
	\|\WW_{\RRR_{j}}\|_2^2 \geq \sum_{i=1}^j \|\WW_{\RRR_{i}}\|_2^2-\|\WW_{\RRR_{i-1}}\|_2^2,
	\end{align*}
	for each $j\geq 1$ we conclude that the process terminates definitely after a finite number of steps and $m\leq k^4/\varepsilon^2$. The partition $\PPP$ satisfies 	\begin{align*}
	\|\WW-\WW_\PPP\|_{\square\QQ} \leq \varepsilon
	\end{align*}
	for each $\QQ$ with $t_\QQ\leq \max\{t, m_0(\varepsilon)\}$ by the choice of $m$ and the construction of the partition sequence, we are left to verify the upper bound on $t_\PPP$ in the statement of the lemma.  
	
	We know that $t_{\RRR_0}=1$, and if the partition does not terminate before the first step, then we assume $m_0 \leq t_{\RRR_1}$. This lower bound does not affect generality, the partition $\QQ_0$ that certifies that $\RRR_0$ is not suitable for the choice of the partition $\PPP$ in the statement of the lemma is selected to have $t_{\QQ_0}=m_0$. For this particular choice of $\QQ_0$ and $\RRR_1$ we can reformulate the condition that the terminating partition $\RRR_m$ has to fulfill as 
	\begin{align*}
	\|\WW-\WW_{\RRR_{m}}\|_{\square\QQ} \leq \varepsilon
	\end{align*} 
	for every $\QQ$ partition of $[0,1]$ into at most $t_{\RRR_{m}}$ classes, since $t_{\RRR_i} \geq m_0$ for every $i \geq 1$. 
	
	We set $s(0)=1$, $s(1)=4m_0$, and further define $s(i+1)=4 s(i)^2$ for each $i \geq 1$. We claim that for each $i\geq 0$ we have $t_{\RRR_i}\leq s(i)$, this can be easily verified by induction, since at each step $\RRR_{i+1}$ is the coarsest common refinement of two partitions with $t_{\RRR_i}$ classes and two additional sets.
	
	Further, for each $i \geq 1$ we have now $\log 4s(i+1)=2 \log 4s(i)$, therefore $s(i)=\frac{(16m_0)^{2^{i-1}}}{4}$, and consequently $s(m)\leq\frac{(16m_0)^{2^{\frac{k^4}{\varepsilon^2}-1}}}{4}$.

	The case regarding $\WW=\WW_\GG$ for a $k$-colored directed graph $\GG$ of vertex cardinality $n$ follows completely identically, at each step of the construction of the partitions $\RRR_i$ the partition $\QQ_0$ and the sets $S$ and $T$ can be chosen to be an $\mathcal{I}_n$-partition and sets, respectively. Hence, $\PPP$ is an $\mathcal{I}_n$-partition, the upper bound on $t_\PPP$ is identical to the one in the general case.
	
	In a similar way we can achieve that $\PPP$ is an equiv-partition, or a $\mathcal{I}_n$-partition with classes of almost equal size in the graph case respectively. Fix $\varepsilon>0$ and a $\WW \in \WWW^{(k)}$. For this setup we define the partition sequence somewhat differently, in particular each element is an equiv-partition.  Set $\RRR_0=[0,1]$ and each $\RRR_{i+1}$  refines the preceding $\RRR_i$. 
	
	If for $i \geq 1$ there exists a partition $\QQ=(Q_1, \dots, Q_{t_{\QQ}})$ of $[0,1]$ into at most $\max\{t_{\RRR_{i-1}}, m_0\}$ classes such that 
	\begin{align}\label{ch3:eq001}
	\|\WW-\WW_{\RRR_{i-1}}\|_{\square\QQ} > \varepsilon, 
	\end{align}
	then we proceed to the construction of $\RRR_i$, otherwise we stop, as above, and set $\PPP=\RRR_{i-1}$, and $m=i-1$. Assume that we are facing the first case. Let $\RRR_{i-1}'$ be the coarsest common refinement of $\RRR_{i-1}$, $\QQ$, and $\{S,T\}$, where the sets $S$ and $T$ certify (\ref{ch3:eq001}) as above. Then 	$\|\WW_{\RRR'_{i-1}}\|_2^2-\|\WW_{\RRR_{i-1}}\|_2^2 > \varepsilon^2/k^4$. Let $\RRR''_{i-1}=(R^2_1, \dots, R^2_l)$ be the partition that is obtained from the classes of $\RRR'_{i-1}=(R^1_1, \dots, R^1_l)$, such that the measure of each of the classes of $\RRR''_{i-1}$ is an integer multiple of $\varepsilon^2/(14k^6t_{\RRR'_{i-1}})$ with $\lambda(R^1_i \triangle R^2_i) \leq \varepsilon^2/(14k^6t_{\RRR'_{i-1}})$ for each $i \in [l]$. (We disregard the technical difficulty of $1/\varepsilon^2$ not being an integer to facilitate readability.)
	
	\begin{claim}\label{ch3:claim4}
		For any kernel $W \colon [0,1]^2 \to \R$ and partitions  $\PPP=(P_1, \dots, P_t)$ and $\Ss=(S_1, \dots, S_l)$ we have
		\begin{align}
		\|W_\PPP-W_\Ss\|_1 \leq 7 \sum_{i=1}^l \lambda(P_i \triangle S_i).
		\end{align}
	\end{claim}
	To see this, let $T_i=P_i \cap S_i$, $N_i= P_i \setminus S_i$, and $M_i=S_i \setminus P_i$ for each $i \in [l]$, and let $\TTT_1=(T_1, \dots, T_l, N_1, \dots , N_l)$ and $\TTT_2=(T_1, \dots, T_l, M_1, \dots, M_l)$ be two partitions of $[0,1]$. Then
	\begin{align*}
	\|W_\PPP-W_{\TTT_1}\|_1 &\leq \sum_{i,j=1}^l \left|\,\int\limits_{T_i \times T_j}  \frac{\int_{P_i \times P_j}W}{\lambda(P_i) \lambda(P_j)} - \frac{\int_{T_i \times T_j}W}{\lambda(T_i) \lambda(T_j)} \right| + 2\sum_{i=1}^l \lambda(N_i) \\
	&\leq \sum_{i,j=1}^l \frac{1}{{\lambda(P_i) \lambda(P_j)\lambda(T_i) \lambda(T_j)}}\Bigg|\int\limits_{T_i \times T_j}  \Bigg[ \lambda(T_i) \lambda(T_j)\left( \, \int\limits_{N_i \times T_j}W+\int\limits_{T_i \times N_j}W+\int\limits_{N_i \times N_j}W \, \right) \\ &\qquad -  \left(\lambda(N_i) \lambda(T_j)+\lambda(T_i) \lambda(N_j)+\lambda(N_i) \lambda(N_j)\right)\int\limits_{T_i \times T_j}W \, \Bigg] \Bigg| + 2\sum_{i=1}^l \lambda(N_i) \\
	&\leq  \sum_{i,j=1}^l 2\|W\|_\infty \frac{ \lambda^2(T_i) \lambda^2(T_j)\left[\lambda(N_i) \lambda(T_j)+\lambda(T_i) \lambda(N_j)+\lambda(N_i) \lambda(N_j)\right]}{{\lambda(P_i) \lambda(P_j)\lambda(T_i) \lambda(T_j)}} \\ &\qquad + 2\sum_{i=1}^l \lambda(N_i) \\
	&\leq 2  \sum_{i,j=1}^l \lambda(N_i)\lambda(P_j) +\lambda(N_j)\lambda(P_i) + 2\sum_{i=1}^l \lambda(N_i) \\
	&= 6 \sum_{i=1}^l \lambda(N_i).
	\end{align*}
	Similarly, 
	\begin{align*}
	\|W_\Ss-W_{\TTT_2}\|_1 \leq 6 \sum_{i=1}^l \lambda(M_i), 
	\end{align*}
	and also
	\begin{align*}
	\|W_{\TTT_2}-W_{\TTT_1}\|_1 \leq 2 \sum_{i=1}^l \lambda(M_i), 
	\end{align*}
	which implies the claim.

	By \Cref{ch3:claim4} it follows that 
	\begin{align*}
	\left| \|\WW_{\RRR'_{i-1}}\|_2^2-\|\WW_{\RRR''_{i-1}}\|_2^2 \right| &=\left|\sum_{\alpha,\beta=1}^{k} \int_{[0,1]^2} (W^{(\alpha, \beta)}_{\RRR'_{i-1}})^2(x,y)-(W^{(\alpha, \beta)}_{\RRR''_{i-1}})^2(x,y) \du x\du y \right|\\
	&\leq \sum_{\alpha,\beta=1}^{k} \left|\int_{[0,1]^2} (W^{(\alpha, \beta)}_{\RRR'_{i-1}}(x,y)-W^{(\alpha, \beta)}_{\RRR''_{i-1}}(x,y))(W^{(\alpha, \beta)}_{\RRR'_{i-1}}(x,y)+W^{(\alpha, \beta)}_{\RRR''_{i-1}}(x,y)) \du x\du y\right| \\
	&\leq  \sum_2 \|\WW\|_\infty {\alpha,\beta=1}^{k} \| W^{(\alpha, \beta)}_{\RRR'_{i-1}}-W^{(\alpha, \beta)}_{\RRR''_{i-1}}\|_1 \\
	&\leq 28 k^2 \sum_{i=1}^l \lambda(R^1_i \triangle R^2_i) \\
	&\leq \varepsilon^2/(2k^4).
	\end{align*}
	We finish with the construction of $\RRR_i$ by refining $\RRR''_{i-1}$ into $\varepsilon^2/(28k^6t_{\RRR'_{i-1}})$ sets in total of equal measure so that the resulting partition refines $\RRR''_{i-1}$. It follows that
	\begin{align}
	\|\WW_{\RRR_{i}}\|_2^2-\|\WW_{\RRR''_{i-1}}\|_2^2 \geq 0,
	\end{align}
	hence
	\begin{align}
	\|\WW_{\RRR_{i}}\|_2^2-\|\WW_{\RRR_{i-1}}\|_2^2 \geq \varepsilon^2/(2k^4).
	\end{align}
	
	The construction of the partitions terminates after at most $2k^4/\varepsilon^2$ steps. The partition $\PPP$ satisfies the norm conditions of the lemma, we are left to check whether it has the right number of classes. Similarly as above, let $s(0)=1$ and $s(1)=m_0(3k)^6/\varepsilon^2$, and further for $i \geq 1$ let $s(i+1)=(3k)^6 s^2(i)/\varepsilon^2$. It is clear from the construction that $t_{\RRR_i}\leq s(i)$. Let $a=(3k)^6/\varepsilon^2$, then it is not difficult to see that $s(i)=\frac{(a^2m_0)^{2^{i-1}}}{a}$. It follows that $t_\PPP \leq \frac{((3k)^{12}m_0/\varepsilon^4)^{2^{2k^4/\varepsilon^2}}}{(3k)^6/\varepsilon^2}$.
	
	The graph case also follows analogously to the general partition case we dealt with above.

\end{proof}
As seen in the proof, the upper bound on the number of classes in the statement of the lemma is not the sharpest we can prove, we stay with the simpler bound for the sake of readability. In the simple graph and graphon case the above reads as follows.

\begin{corollary}
	\label{ch3:sreg}
	For every $\varepsilon>0$ and $W \in \mathcal W$ there exists a partition $\PPP=(P_1, \dots, P_m)$ of $[0,1]$ into $m \leq 16^{2^{1/\varepsilon^2}}/4$ parts, such that 
	\begin{align}
	\|W-W_\PPP\|_{\square\QQ} \leq \varepsilon. 
	\end{align}
	for each partition $\QQ$ of $[0,1]$ into at most $t_\PPP$ classes.
	
	With the additional condition that the partition classes should have the same measure the above is true with $m \leq \frac{(3^{12}/\varepsilon^4)^{2^{(2^4/\varepsilon^2)}}}{3^6/\varepsilon^2}$.
\end{corollary}

%
%
%
%
%

\section{Proof of Theorem \ref{ch6:graphmain}}\label{sec.thm}


We will exploit the continuity of a testable graph parameter with respect to the cut norm and distance, and the connection of this characteristic to the sample complexity of the parameter. We require two results, the first one quantifies the above continuity. We generally assume that the sample complexity satisfies $q_g(\varepsilon)\geq 1/\varepsilon$, also $\varepsilon \leq 1$ and $k\geq 2$.
\begin{lemma} \label{ch6:dev}
	Let $g$ be a testable $k$-colored digraph parameter with sample complexity at most $q_g$. Then for any $\varepsilon > 0$ and two graphs, $\GG$ and $\HHH$, with $|V(\GG)|,|V(\HHH)| \geq \left(\frac{2q_g^2(\varepsilon/4)}{\varepsilon}\right)^{1/(q_g(\varepsilon/4)-1)}$ satisfying $\delta_\square(\GG,\HHH) \leq k^{- 2q_g^2(\varepsilon/4)}$ we have \begin{align*}|g(\GG)-g(\HHH)| \leq \varepsilon. \end{align*}
\end{lemma}
\begin{proof}
	Let $\varepsilon>0$, $\GG$ and $\HHH$ be as in the statement, and set $q=q_g(\varepsilon/4)$. Then we have
	\begin{align} \label{ch6:1eq1}
	|g(\GG)-g(\HHH)|  & \leq |g(\GG)-g(\G(q,\GG))|+ 
	|g(\G(q,\WW_\GG))-g(\G(q,\GG))| \nonumber \\
	&  \quad + |g(\G(q,\WW_\GG))-g(\G(q,\WW_\HHH))| + |g(\G(q,\HHH))-g(\G(q,\WW_\HHH))| \nonumber \\ & \quad +|g(\HHH)-g(\G(q,\HHH))|.
	\end{align}
	The first and the last term on the right of $(\ref{ch6:1eq1})$ can be each upper bounded by $\varepsilon/4$ with cumulative failure probability $\varepsilon/2$ due to the assumptions of the lemma. To deal with the second term we require the fact that $\G(q, \GG)$ and $\G(q, \WW_\GG)$ have the same distribution conditioned on the event that the $X_i$ variables that define $\G(q,\WW_\GG)$ lie in different classes of the canonical equiv-partition of $[0,1]$ into $|V(\GG)|$ classes. The failure probability of the latter event can be upper bounded by $q^2/2|V(\GG)|^{q-1}$, which is at most $\varepsilon/4$, analogously for the fourth term. Until this point we have not dealt with the relationship of the two random objects $\G(q, \WW_\GG)$ and $\G(q, \WW_\HHH)$, therefore the above discussion is valid for every coupling of them. 
	
	In order to handle the third term we upper bound the probability that the two random graphs are different by means of an appropriate coupling, since clearly in the event of identity the third term of (\ref{ch6:1eq1}) vanishes. More precisely, we will show that $\G(q,\WW_\GG)$ and $\G(q,\WW_\HHH)$ can be coupled in such a way that $\PP(\G(q,\WW_\GG)\neq \G(q,\WW_\HHH)) < 1-\varepsilon$. We utilize that for a fixed $k$-colored digraph $\FFF$ on $q$ vertices we can upper bound the deviation of the subgraph densities of $\FFF$ in $\GG$ and $\HHH$  through the cut distance of these graphs, see \Cref{ch3:countlemma}. In particular, 
	
	\begin{align}
	|\PP(\G(q,\WW_\GG)=\FFF)-\PP(\G(q,\WW_\HHH)=\FFF)| \nonumber \leq {q \choose 2} \delta_\square(\WW_\GG,\WW_\HHH) .
	\end{align}
	Therefore in our case
	\begin{align*}
	\sum_{\FFF} |\PP(\G(q,\WW_\GG)=\FFF)-\PP(\G(q,\WW_\HHH)=\FFF)| \leq k^{2{q \choose 2 }} {q \choose 2 } k^{-2q^2} \leq \varepsilon,
	\end{align*}
	where the sum goes over all labeled $k$-colored digraphs $\FFF$ on $q$ vertices.
	
	Since there are only finitely many possible target graphs for the random objects, we can couple $\G(q,\WW_\GG)$ and $\G(q,\WW_\HHH)$ 
	so that in the end we have $\PP(\G(q,\WW_\GG)\neq \G(q,\WW_\HHH)) \leq \varepsilon$. This implies that with positive probability (in fact, with at least $1-2\varepsilon$) the sum of the five terms on the right hand side of (\ref{ch6:1eq1}) does not exceed $\varepsilon$, so the statement of the lemma follows.
	
\end{proof} 

We will also require the following statement which can be regarded as the quantitative counterpart of Lemma 3.2 from \cite{LV}. It clarifies why the cut-$\PPP$-norm, (\Cref{ch3:cutpnorm}, \Cref{ch3:cutpnorm2}) and  the need for the accompanying regularity lemma, \Cref{ch3:sregk}, are essential for our intent. 

\begin{lemma}\label{ch6:lemma1}
	Let $k\geq 2$, $\varepsilon>0$, $U$ be a step function with steps $\PPP=(P_1, \dots, P_t)$ and $V$ be a graphon with $\|U-V\|_{\square\PPP} \leq \varepsilon.$ For any $k$-colored digraphon $\UU=(U^{(1,1)},\dots,U^{(k,k)})$ that is a step function with steps from $\PPP$ and a $(k,m)$-coloring of $U$  there exists a $(k,m)$-coloring $\VV=(V^{(1,1)},\dots,V^{(k,k)})$ of $V$ so that $\|\UU-\VV\|_\square= \sum_{\alpha, \beta=1}^k \|U^{(\alpha, \beta)}-V^{(\alpha, \beta)}\|_\square \leq k^2\varepsilon.$
	
	If $V=W_G$ for a simple graph $G$ on $n\geq 16/\varepsilon^2$ nodes and $\PPP$ is an $\mathcal{I}_n$-partition of $[0,1]$ then there is a $(k,m)$-coloring $\GG$ of $G$ that satisfies the above conditions and $\|\UU-\WW_\GG\|_\square \leq 2k^2\varepsilon.$
	
\end{lemma}
\begin{proof}
	Fix $\varepsilon >0$, and let $U$, $V$, and $\UU$ be as in the statement of the lemma. Then $\sum_{\alpha,\beta=1}^k U^{(\alpha, \beta)}=1$, let $M$ be the subset of $[k]^2$ such that its elements have at least one component that is at most $m$, so we have  $\sum_{(\alpha, \beta) \in M} U^{(\alpha, \beta)}=U$ by definition. For $(\alpha, \beta) \in M$ set $V^{(\alpha, \beta)}= \frac{V U^{(\alpha, \beta)}}{U}$ on the set where $U>0$ and $V^{(\alpha, \beta)}=\frac{V}{k^2-(k-m)^2}$ where $U=0$, furthermore for $(\alpha, \beta) \notin M$ set $V^{(\alpha, \beta)}= \frac{(1-V) U^{(\alpha, \beta)}}{1-U}$ on the set where $U<1$ and $V^{(\alpha, \beta)}=\frac{1-V}{(k-m)^2}$ where $U=1$. We will show that the $k$-colored digraphon $\VV$ defined this way satisfies the conditions, in particular for each $(\alpha, \beta) \in [k]^2$ we have $\|U^{(\alpha, \beta)}-V^{(\alpha, \beta)}\|_\square \leq \varepsilon $. We will explicitly perform the calculation only for $(\alpha, \beta) \in M$, the other case is analogous. Fix some $S, T \subset [0,1]$, then
	\begin{align*}
	&\left|\int_{S\times T}  U^{(\alpha, \beta)}-V^{(\alpha, \beta)}\right| = \left|\int_{S\times T, U>0}  U^{(\alpha, \beta)}-V^{(\alpha, \beta)} + \int_{S\times T,U=0}  U^{(\alpha, \beta)}-V^{(\alpha, \beta)}\right|  \\
	& \qquad\leq \sum_{i,j=1}^t \left| \int_{(S\cap P_i) \times (T\cap P_j), U>0} \frac{U^{(\alpha, \beta)}}{U} (U-V) + \int_{(S\cap P_i) \times (T\cap P_j),U=0} \frac{1}{k^2-(k-m)^2} (U-V)\right| \\
	&\qquad =\sum_{i,j=1}^t \left|\int_{(S\cap P_i) \times (T\cap P_j)} (U - V ) \left[\I_{U>0} \frac{U^{(\alpha, \beta)}}{U}+ \I_{U=0}\frac{1}{k^2-(k-m)^2}\right]\right|\\ & \qquad\leq \sum_{i,j=1}^t \left|\int_{(S\cap P_i) \times (T\cap P_j)} (U - V )\right| \\ & \qquad = \|U-V\|_{\square\PPP} \leq \varepsilon.
	\end{align*}
	The second inequality is a consequence of $\left[\I_{U>0} \frac{U^{(\alpha, \beta)}}{U}+ \I_{U=0}\frac{1}{k^2-(k-m)^2}\right]$ being a constant between $0$ and $1$ on each of the rectangles $P_i \times P_j$. 
	
	We prove now the second statement of the lemma concerning graphs with $V=W_G$ and a partition $\PPP$ that is an $\mathcal{I}_n$-partition. The general discussion above delivers the existence of $\VV$ that is a $(k,m)$-coloring of $W_G$, which can be regarded as a fractional coloring of $G$, as $\VV$ is constant on the sets associated with nodes of $G$. For $|V(G)|=n$ we get for each $ij \in {[n] \choose 2}$ a probability distribution on $[k]^2$ with $\PP\left(Z_{ij}=(\alpha,\beta)\right)=n^2 \int_{[\frac{i-1}{n},\frac{i}{n}] \times [\frac{j-1}{n},\frac{j}{n}]} V^{(\alpha, \beta)}(x,y) \du x \du y$. For each pair $ij$ we make an independent random choice according to this measure, and color $(i,j)$ by the first, and $(j,i)$ by the second component of $Z_{ij}$ to get a proper $(k,m)$-coloring $\GG$ of $G$. It remains to conduct the analysis of the deviation in the statement of the lemma, we will show that this is small with high probability with respect to the randomization, which in turn implies the existence. 
	We have

	\begin{align*}
	\|\UU-\WW_\GG\|_\square &\leq \|\UU-\VV\|_\square +  \|\VV-\WW_\GG\|_\square \\
	&\leq k^2\varepsilon + \sum_{\alpha,\beta=1}^k   \|V^{(\alpha,\beta)}-W^{(\alpha,\beta)}_\GG\|_\square
	\end{align*}
	For each  $(\alpha, \beta) \in [k]^2$ we have that $\PP\left( \|V^{(\alpha,\beta)}-W^{(\alpha,\beta)}_\GG\|_\square \geq 4/\sqrt{n} \right) \leq 2^{-n}$, this result is exactly Lemma 4.3 in \cite{BCL}. This implies for $n \geq 16/\varepsilon^2$ the existence of a suitable coloring, which in turn finishes the proof of the lemma.
	
\end{proof}

\begin{remark}
	Actually we can perform the same proof to verify the existence of a $k$-coloring $\VV$ such that $d_{\square\PPP}(\UU, \VV) \leq k^2\varepsilon$. On the other hand, we can not weaken the condition on the closeness of $U$ and $V$, a small cut-norm of $U-V$ does not imply the existence of a suitable coloring $\VV$, for example in the case when the number of steps of $\UU$ is exponential in $1/\|U-V\|_{\square}$.
\end{remark}

We proceed towards the proof of the main statement of the paper. Before we can outline that we require yet another specific lemma.

Let $\mathcal M_{\Delta,n}$ denote the set of $\mathcal{I}_n$-step functions $U$ that have steps $\PPP_U$ with $|\PPP_U| \leq t_k(\Delta,1)$ classes,  and values between $0$ and $1$, where $t_k$ is the function from \Cref{ch3:sregk}. In order to verify \Cref{ch6:graphmain} we will condition on the event that is formulated in the following lemma. Recall \Cref{ch3:cutpnorm} for the deviation $d_{W,\PPP}(V)$.

\begin{lemma}\label{ch6:lemma2}
	Let $G$ be a simple graph on $n$ vertices and $\Delta >0$. Then for $q \geq  
	2^{2^{(2k^4/\Delta^2)+4}}$ we have 
	\begin{align}
	| d_{U, \PPP_U}(G) - d_{U, \PPP_U}(\G(q,G))| \leq \Delta, 
	\end{align}
	for each $U \in \mathcal M=\mathcal M_{\Delta,n}$ simultaneously with probability at least $1-\exp(-\frac{\Delta^2q}{2^7})$, whenever $n \geq 4q/\Delta$. 
\end{lemma}

\begin{proof}
	Let $G$ and $\Delta>0$ be arbitrary, and $q$ be such that it satisfies the conditions of the lemma. For technical convenience we assume that $n$ is an integer multiple of $q$, let us introduce the quantity $t_1=t_k(\Delta,1)=2^{2^{\frac{k^4}{\Delta^2}+2}-2}$, and denote $\G(q,G)$ by $F$. For the case when $q$ is not a divisor of $n$, then we just add at most $q$ isolated vertices to $G$ to achieve the above condition, by this operation $d_{U, \PPP}(G)$ is changed by at most $q/n$. Also, we can couple in a way such that $d_{U, \PPP}(\G(q,G))$ remains unchanged with  probability at least $1-q/n.$
	
	We will show that there exists an $\mathcal{I}_n$-permutation $\phi$ of $[0,1]$ such that $\|W_G-W_F^\phi\|_{\square\QQ} < \Delta$ for any $\mathcal{I}_n$-partition $\QQ$ of $[0,1]$ into at most $t_1$ classes with high probability. Applying \Cref{ch3:sregk} with the error parameter $\Delta/4$ and $m_0(\Delta)=t_k(\Delta,1)$ for approximating $W_G$ by a step function we can assert that there exists an $\mathcal{I}_n$-partition $\PPP$ of $[0,1]$ into $t_\PPP$ classes with $t_\PPP \leq t_2$ with $t_2=t_k(\Delta/4,t_k(\Delta,1))=2^{2^{(2k^4/\Delta^2)+2}+2^{(k^4/\Delta^2)+1}}\leq 2^{2^{(2k^4/\Delta^2)+3}}$ such that for every  $\mathcal{I}_n$-partition $\QQ$ into $t_\QQ$ classes  $ t_\QQ \leq \max\{t_\PPP, t_1\}$ it holds that
	\begin{align*}
	\|W_G-(W_G)_\PPP\|_{\square \QQ} \leq \Delta/4.
	\end{align*}
	We only need here 
	\begin{align}\label{ch6:eq111}
	\sup_{\QQ : t_\QQ\leq t_1} \|W_G-(W_G)_\PPP\|_{\square \QQ} \leq \Delta/4.
	\end{align}
	This property is by \Cref{ch3:rem1} equivalent to stating that 
	\begin{align}\label{ch6:1eq3}
	\max_{\QQ}\max_{A \in \mathbb A}\max_{S,T \subset [0,1]} \sum_{i,j=1}^{t_1} A_{i,j} \int_{S \times T} (W_G-(W_G)_\PPP)(x,y) \I_{Q_i}(x)\I_{Q_j}(y) \du x \du y \leq \Delta/4, 
	\end{align}
	where $\mathbb A$ is the set of all $t_1 \times t_1$ matrices with $-1$ or $+1$ entries.
	
	We can reformulate the above expression (\ref{ch6:1eq3}) by putting 
	\begin{align*}J =
	\begin{pmatrix}
	1 & 0 & 1 & 0 \\
	1 & 0 & 1 & 0 \\
	0 & 0  & 0 & 0  \\
	0 & 0 & 0 & 0
	\end{pmatrix},
	\end{align*}
	and defining the tensor product $B_A=A \otimes J$, so that $B_{i,j}^{\alpha, \beta}=A_{ij}J_{\alpha, \beta}$ for each $A \in \mathbb A$. The first matrix $J$ corresponds to the $\mathcal{I}_n$-partition $(S\cap T, S \setminus T, T\setminus S, [0,1]\setminus(S \cup T))=(T_1, T_2, T_3, T_4)$ generated by a pair $(S,T)$ of $\mathcal{I}_n$-sets of $[0,1]$ so that for any function $U \colon[0,1]^2 \to \R$ it holds that
	\begin{align*}
	\sum_{i,j=1}^4 J_{ij} \int_{[0,1]^2} U(x,y) \I_{T_i}(x) \I_{T_j} \du x \du y = \int_{S \times T} U(x,y) \du x \du y.
	\end{align*}
	
	It follows that the inequality (\ref{ch6:1eq3}) is equivalent to saying
	\begin{align}\label{ch6:1eq4}
	\max_{A \in \mathbb A}\max_{\hat\QQ} \sum_{i,j=1}^{t_1}\sum_{\alpha, \beta=1}^{4} (B_A)^{\alpha,\beta}_{i,j} \int_{[0,1]^2} (W_G-(W_G)_\PPP)(x,y) \I_{Q_i^\alpha}(x)\I_{ Q_j^\beta}(y) \du x \du y \leq \Delta/4, 
	\end{align} 
	where the second maximum goes over all $\mathcal{I}_n$-partitions $\hat \QQ= (Q_i^\alpha)_{\substack{i \in [t_1]\\ \alpha \in [4]}}$ into $4t_1$ classes. Let us substitute an arbitrary graphon $U$ for $W_G-(W_G)_\PPP$ in (\ref{ch6:1eq4}) and define 
	\begin{align*}\hat h_{A,\hat \QQ}(U)= \sum_{\substack{1\leq i,j\leq t_1\\1\leq \alpha,\beta \leq 4 }} (B_A)^{\alpha,\beta}_{i,j} \int_{[0,1]^2} U(x,y) \I_{Q_i^\alpha}(x)\I_{ Q_j^\beta}(y) \du x \du y\end{align*}
	and
	
	\begin{align*}\hat h_{A}(U)=\max_{\hat\QQ} h_{A,\hat \QQ}(U)\end{align*} as the expression whose optima is sought for a fixed $A \in \mathbb A$. 
	
	For notational convenience only lower indices will be used when referring to the entries of $B_A$. We introduce a relaxed version $h_{A}$ of the above function  $\hat h_{A}$ by replacing the requirement on $\hat \QQ$ being an $\mathcal{I}_n$-partition, instead we define \begin{align*}h_{A,f}(U)=\sum_{1\leq i,j\leq 4t_1} (B_A)_{i,j} \int_{[0,1]^2} U(x,y) f_i(x)f_j(y) \du x \du y\end{align*} with $f=(f_i)_{i \in [4t_1]}$ being a fractional $\mathcal{I}_n$-partition into $4t_1$ classes, that is, each component of $f$ is a non-negative $\mathcal{I}_n$-function, and their sum is the constant $1$ function. Further, we define $h_A(U)=\max_f h_{A,f}(U)$, where $f$ runs over all fractional $\mathcal{I}_n$-partitions into $4t_1$ parts. It is easy to see that 
	\begin{align*}|\hat h_{A}(U)- h_{A}(U)|\leq 1/n,\end{align*}
	since the two functions coincide when $U$ is $0$ on the diagonal blocks.
	Denote $U'=W_{\hh(q,U)}$, where the graphon is given by the increasing order of the sample points $\{ \, X_i\mid i \in [q]\,\}$. We wish to upper bound the probability that the deviation $| h_A(U)- h_A(U')|$ exceeds $\Delta/4$, for some $A \in \mathbb A$. Similarly as above, $|\hat h_{A}(U')- h_{A}(U')|\leq 1/q.$
	
	We remark, that a simple approach would be using a slight variant of the counting lemma \Cref{ch3:countlemma} that $|h_A(U)-h_A(U')|\leq 16t_1^2 \delta_\square(U,U')$ together with a version of \Cref{ch3:samp} for kernels with perhaps negative values, this way we would have to impose a lower bound on $q$ that is exponential in $t_1$ in order to satisfy the statement of the lemma. We can do slightly better using more involved methods.
	
	We require the notion of ground state energies from \cite{BCL2}. Let
	
	\begin{align*}
	\hat \EEE(G,J)=\max_{\QQ} \sum_{i,j=1}^s J_{i,j} \int_{[0,1]^2} \I_{Q_{i}}(x)\I_{Q_{j}}(y) W_G(x,y) \du x \du y,
	\end{align*}
	where the maximum runs over all $\mathcal{I}_n$-partitions $\QQ$ into $s$ parts when $|V(G)|=n$.
	
	Further,
	\begin{equation*} 
	\EEE(U,J)=\sup_{f} \sum_{i,j=1}^s J_{i,j} \int_{[0,1]^2} f_{i}(x)f_{j}(y) U(x,y) \du x \du y, 
	\end{equation*}	
	where the supremum runs over all fractional partitions $f$ into $s$ parts.

	The next result was first proved in \cite{AVKK2}, subsequently refined in \cite{KM1}. 

	\begin{theorem}\cite{AVKK2}\cite{KM1}\label{ch4:maincor1}
		Let $s \geq 1$, and $\rho >0$. Then there is an absolute constant $c>0$ such that for any $s \geq 1$, $\rho >0$, kernel $U$, real matrix $J$, and $q \geq c\Theta^4  \log(\Theta)$ with $\Theta=\frac{s^2}{\rho}$  we have 
		\begin{align}\label{ch6:1eqq}
		\PP(|\EEE(U,J)-\hat \EEE(\G(q,U),J)|>\rho \|U\|_\infty)< 2 \exp\left( -\frac{\rho^2q}{32}\right).
		\end{align}
	\end{theorem}

	We have seen above that $\hat h_{A}(U)=\hat \EEE(U,B_A)$ and $h_{A}(U)=\EEE(U,B_A).$ Since $q\geq t_2^2 \geq 2^{85} t_1^{10}/\Delta^5$ we can apply \Cref{ch4:maincor1} for each $A \in \mathbb A$ with $s=4t_1$, and $\rho=\Delta/4.$

	This shows eventually that with probability at least $1- 2^{16t_1^2+1} \exp\left( -\frac{\Delta^2q}{2^9}\right)\leq 1- \exp\left( -\frac{\Delta^2q}{2^8}\right)$ we have 
	\begin{align}\label{ch6:1eq5}
	\max_{A \in \mathbb A}\max_{\hat\QQ} \sum_{i,j=1}^{t_1}\sum_{\alpha, \beta=1}^{4} (B_A)^{\alpha,\beta}_{i,j} \int_{[0,1]^2} (W_{\G(q, U)})(x,y) \I_{Q_i^\alpha}(x)\I_{ Q_j^\beta}(y) \du x \du y \leq \Delta/2,
	\end{align} 
	where the second maximum runs over all $\mathcal{I}_q$-partitions $\hat\QQ$ of $[0,1]$ into $4t_1$ parts. Denote this event by $E_1$.
	
	This however is equivalent to saying that for every $\QQ$ partition into $t_\QQ$ classes  $ t_\QQ \leq t_1$ it is true that
	\begin{equation}\label{ch6:1eq6}
	\|W_{\G(q, U)}\|_{\square \QQ} \leq \Delta/2.
	\end{equation}

	The second estimate we require concerns the closeness of the step function $(W_G)_\PPP$ and its sample $W_{\hh(q,(W_G)_\PPP)}$. Our aim is to overlay these two functions via measure preserving permutations of $[0,1]$, such that the measure of the subset of $[0,1]^2$ where they differ is as small as possible. 
	
	Let $V=W_{\hh(q,(W_G)_\PPP)}$, this $\mathcal{I}_n$-function is well-defined this way and is a step function with steps forming the $\mathcal{I}_n$-partition $\PPP'$. This latter $\mathcal{I}_n$-partition of $[0,1]$ is the image of $\PPP$ induced by the sample $\{X_1, \dots, X_q\}$ and the map $i \mapsto [\frac{i-1}{q}, \frac{i}{q})$. Let $\psi$ be a measure preserving $\mathcal{I}_n$-permutation of $[0,1]$ that satisfies that for each $i \in [t_\PPP]$ the volumes $\lambda(P_i \triangle \psi(P'_i))=|\lambda(P_i) - \lambda(P'_i)|$. Let $\PPP''$ denote the partition with classes $P''_i=\psi(P'_i)$ and $V'=(V)^\psi$ (note that $V'$ and $V$ are equivalent as graphons), furthermore let $N$ be the (random) subset of $[0,1]^2$ where the two functions $(W_G)_\PPP$ and $V'$ differ. Then  
	\begin{align}
	\E[\lambda(N)] \leq 2 \E[\sum_{i=1}^{t_\PPP}|\lambda(P_i) - \lambda(P'_i)|].
	\end{align}
	The random variables $\lambda(P'_i)$ for each $i$ can be interpreted as the proportion of positive outcomes out of $q$ independent Bernoulli trials with success probability $\lambda(P_i)$. By Cauchy-Schwarz it follows that 
	\begin{align}
	\E[\sum_{i=1}^{t_\PPP}|\lambda(P_i) - \lambda(P'_i)|] \leq \sqrt{t_\PPP \E[\sum_{i=1}^{t_\PPP}(\lambda(P_i) - \lambda(P'_i))^2]} \leq \sqrt{\frac{t_2}{q}}.
	\end{align}
	
	This calculation yields that $\E[\lambda(N)] \leq  \sqrt{\frac{4t_2}{q}} \leq \Delta/8$ by the choice of $q$, since $q \geq t_2^2$. Standard concentration result gives us that $\lambda(N)$ is also small in probability if $q$ is chosen large enough. For convenience, define the martingale $M_l=\E[\lambda(N)|X_1, \dots, X_l]$ for $1 \leq l\leq q$, and notice that the martingale differences are uniformly bounded, $|M_l-M_{l-1}| \leq \frac{4}{q}$. The Azuma-Hoeffding inequality then yields 
	\begin{align}\label{ch6:event2}
	\PP(\lambda(N) \geq \Delta/4 ) &\leq \PP(\lambda(N)\geq \E[\lambda(N)] + \Delta/8 ) \leq \exp(-\Delta^2q/2^{11}).
	\end{align} 
	
	Define the event $E_2$ that holds whenever $\lambda(N) \leq \Delta/4$, and condition on $E_1$ and $E_2$, the failure probability of each one is at most $\exp(-\frac{\Delta^2q}{2^{11}})$.
	
	It follows that $\|V^\psi-(W_G)_\PPP\|_1 \leq \lambda(N) \leq \Delta/4$. Now employing the triangle inequality and the bound (\ref{ch6:1eq6}) we get for all $\mathcal{I}_n$-partitions $\QQ$ into $t$ parts that 
	\begin{align*}
	\|W_G-(W_F)^\psi\|_{\square \QQ} &\leq \|W_G-(W_G)_\PPP\|_{\square \QQ}+\|(W_G)_\PPP-V^\psi \|_1+ \|V^\psi-(W_F)^\psi\|_{\square \psi(\QQ)} \leq \Delta.
	\end{align*}
	
	Now let $U \in \mathcal M_{\Delta,n}$ be arbitrary, and let $\PPP_U$ denote the partition consisting of the steps of $U$. Let $\phi$ be the $\mathcal{I}_n$-permutation of $[0,1]$ that is optimal in the sense that $d_{U, \PPP_U} (G)=\|U-(W_G)^\phi\|_{\square\PPP_U}$. Then 
	\begin{align*}
	d_{U, \PPP_U} (G) - d_{U, \PPP_U} (F) &\leq \|U-(W_G)^\phi\|_{\square\PPP_U} - \|U-(W_F)^{(\psi \circ \phi)}\|_{\square\PPP_U} \\
	&\leq \|W_G-(W_F)^\psi\|_{\square \phi^{-1}(\PPP_U)} \leq \Delta.
	\end{align*}
	The lower bound on the above difference can be handled in a similar way, therefore we have that $|d_{U, \PPP_U} (G) - d_{U, \PPP_U} (F)|\leq \Delta$ for every $U \in \mathcal M_{\Delta,n}$.

	We conclude the proof with mentioning that the failure probability of  the two events $E_1$ and $E_2$ taking place simultaneously is at most $\exp(-\frac{\Delta^2q}{2^7})$. 
	
\end{proof}

We are now ready to conduct the proof of the main result of the paper concerning graph parameters.

\begin{proofof}{\Cref{ch6:graphmain}}

	Let us fix $\varepsilon>0$ and the simple graph $G$ with $n$ vertices. We introduce the error parameter  $\Delta=\frac{k^{-2q_g^2(\varepsilon/2) }}{4k^2+1}$ and set $q\geq 
	2^{2^{(2k^4/\Delta^2)+4}}$. To establish the lower bound on $f(\G(q,G))$  not  much effort is required: we pick a $(k,m)$-coloring $\GG$ of $G$ that certifies the value $f(G)$, that is, $g(\GG)=f(G)$. Then the $(k,m)$-coloring of $\FFF=\G(q, \GG)$ of $\G(q,G)$  induced by $\GG$ satisfies $g(\FFF)\geq g(\GG)-\varepsilon/2$ with probability at least $1-\varepsilon/2$ since $q \geq q_g(\varepsilon/2)$, due to the testability property of $g$, which in turn implies $f(\G(q,G)) \geq f(G)-\varepsilon/2$ with probability at least $1-\varepsilon/2$.
	
	The problem concerning the upper bound in terms of $q$ on $f(\G(q,G))$ is the difficult part of the proof, the rest of it deals with this case. Recall that $\mathcal M_{\Delta,n}$ denotes the set of the $[0,1]$-valued proper $\mathcal{I}_n$-step functions that have at most $t_k(\Delta,1)$ steps. Let us condition on the event in the statement of \Cref{ch6:lemma2}, that is for all $U \in \mathcal M_{\Delta,n}$ it holds that $| d_{U, \PPP_U}(G) - d_{U, \PPP_U}(\G(q,G))| \leq \Delta.$  Let $\mathcal N$ be the set of all $k$-colored digraphons $\WW$ that are $\mathcal{I}_n$-step functions with at most $t_k(\Delta,1)$ steps $\PPP$, and that satisfy $d_{U,\PPP}(G) \leq 2\Delta$ for $U=\sum_{(\alpha,\beta)\in M} W^{(\alpha,\beta)}$.

	Our main step in the proof is that, conditioned on the aforementioned event, we construct for each $(k,m)$-coloring of $F$ a corresponding coloring of $G$ so that the $g$ values of the two colored instances are sufficiently close. We elaborate on this argument in the following.
	
	Let us fix an arbitrary $(k,m)$-coloring of $F$ denoted by $\FFF$. According to \Cref{ch3:sregk} there exists a $\WW$ that is a proper $\mathcal{I}_n$-step function with at most $t_k(\Delta,1)$ steps $\PPP_\WW$ such that there exists an  $\mathcal{I}_n$-permutation $\phi$ of $[0,1]$ such that $d_{\square\PPP}((\WW_{\FFF})^\phi,\WW) \leq \Delta$.  Therefore, by setting $U=\sum_{(\alpha,\beta)\in M} W^{(\alpha,\beta)}$ we have $d_{U,\PPP}(F) \leq \Delta$ and $U \in \mathcal M_{\Delta,n}$. This in turn implies that  $d_{U,\PPP}(G) \leq 2\Delta$, and consequently  $\WW \in \mathcal N$.
	It follows from \Cref{ch6:lemma1} that there exists a $(k,m)$-coloring of $G$ denoted by $\GG$ such that $d_{\square} (\WW,(\WW_\GG)^\psi) \leq 4k^2\Delta$ for some $\psi$ that is an $\mathcal{I}_n$-permutation of $[0,1]$. 
	
	Therefore we get that $\delta_{\square}(\GG,\FFF) \leq (4k^2+1)\Delta$. By virtue of \Cref{ch6:dev} we can assert that $|g(\GG)-g(\FFF)| \leq \varepsilon/2$. This finishes our argument, as $\FFF$ was arbitrary, and the failure probability of the  conditioned event in the analysis of the upper bound is at most $\varepsilon/2.$

	\qed
\end{proofof}

\section{Weak nondeterminism} \label{ch6:sec:weaknd}

We introduce an even more restrictive notion of nondeterminism corresponding to node colorings (\Cref{ch6:defndtest} used throughout the paper is a special case of the nondeterminism notion used commonly in complexity theory). Relying on this new concept we are able to improve on the  upper bound  of the sample complexity using a simplified version of our approach applied in the proof of \Cref{ch6:graphmain} without significant alterations.

We formulate the definition of a stronger property than the previously defined nondeterministic testability. The notion itself may seem at first more involved, but in fact it only corresponds to the case, where the witness parameter $g$ for a graph $G$ is evaluated only on the set of node-colorings instead of edge-colorings of $G$ in order to determine the $f$ value in the maximum expression. This modification will enable us to rely only on the cut-norm and the corresponding regularity lemmas instead of the cut-$\PPP$-norm that was employed in the general case, thus leads us to improved upper bounds on the sample complexity of $f$ with respect to that of $g$. This time we only treat the case of undirected graph colorings in detail, the directed case is analogous. 

We will introduce the set of colorings of $G$ called  node-$(k,m)$-colorings\sindex[notions]{node-$(k,m)$-coloring}. Let $\mathcal T=(T_1, \dots, T_k)$ be a partition of $V(G)$ and $\mathcal D =((D_1, \dots, D_m),( D_1', \dots, D_m'))$ be two partitions of $[k]^2$, together they induce two partitions, $\mathcal C=((C_1, \dots, C_m), (C_1', \dots, C_m'))$,  of $V(G)^2$ such that each class is of the form $C_i=\cup_{(\alpha,\beta) \in D_i} T_\alpha \times T_\beta$ and $C_i'=\cup_{(\alpha,\beta) \in D_i'} T_\alpha \times T_\beta$ respectively. A node-$(k,m)$-coloring of $G$ is defined by some $\mathcal C$ of the previous form and is the $2m$-tuple of simple graphs $\GG=(G_1, \dots, G_m, \tilde G_1, \dots, \tilde G_m)$ with $G_i=G[C_i]$ and $\tilde G_i= G^c[C_i']$. Here $G^c$ stands for the complement of $G$ (the union of $G$ and its complement is the undirected complete graph), and $G[C_i]$ is the union of induced labeled subgraphs of $G$ between $T_\alpha$ and $T_\beta$ for each $(\alpha,\beta) \in D_i$ for $\alpha \neq \beta$, in the case of $\alpha=\beta$ the term in the union is the induced labeled subgraph of $G$ on the node set $T_\alpha$.

These special edge-$2m$-colored graphs that can serve as node-$(k,m)$-colorings are given by a triple $(G, \mathcal T, \mathcal D)$, where $G$ is a simple graph, $\mathcal T$ is a partition of $V(G)$ into $k$ parts, and $\mathcal D$ is a pair of partitions of $[k]^2$ into $m$ parts. In the case of $r$-uniform hypergraphs for arbitrary $r\geq 2$ a node-$(k,m)$-colorings is also a triple $(G, \mathcal T, \mathcal D)$, the only difference in comparison to the graph case is that $\mathcal D$ is a pair of partitions of $[k]^r$ into $m$ parts, the rest of the description is analogous. 

\begin{definition}
	The $r$-uniform hypergraph parameter $f$ is weakly non-deterministically testable\sindex[notions]{graph parameter!weakly non-deterministically testable} if there exist integers $m$ and $k$ with $m \leq k^r$ and a testable edge-$2m$-colored directed $r$-graph parameter $g$ such that for any simple $r$-graph $G$ we have $f(G)=\max_{\GG} g(\GG)$, where the maximum goes over  the set of node-$(k,m)$-colorings of $G$.
	
\end{definition}

We present two approaches to handle this variant of the nondeterministic testability. The first method follows the proof framework introduced for the general case for graphs, its adaptation to the current setting results an improvement on the upper bound on the sample complexity to a $2$-fold exponential of the sample complexity of the witness a parameter and is also applicable to the corresponding property testing setting. The second idea entails the graph case as well as the $r$-uniform hypergraph setting for arbitrary rank $r$ of the weak setting. We manage to reduce the  upper bound on the sample complexity further to only exponential dependence. This approach does seems to be more problem specific, than the previous one, and it does not directly yield an analogous statement in property testing.

\subsection*{First approach}

The following lemma is the analogous result to \Cref{ch6:lemma1} that can be employed in the proof of the variant of \Cref{ch6:graphmain} for the special case of weakly nondeterministically testable graph parameters. 

\begin{lemma}\label{ch6:lemma3}
	Let $\varepsilon>0$, let $U$ and $V$ be arbitrary graphons with $\|U-V\|_{\square} \leq \varepsilon$, and also let $k \geq 2$ and $m \leq k^2$. For any $\UU=(U^{(1)},\dots,U^{(m)},\tilde U^{(1)},\dots,\tilde U^{(m)} )$ node-$(k,m)$-coloring of $U$ there exists a node-$(k,m)$-coloring of $V$ denoted by $\VV=(V^{(1)},\dots,V^{(k)},\tilde V^{(1)},\dots,\tilde V^{(m)})$ such that $d_\square(\UU,\VV)= \sum_{i=1}^m \|U^{(i)}-V^{(i)}\|_\square +\sum_{i=1}^m \|\tilde U^{(i)}-\tilde V^{(i)}\|_\square \leq 2k^2\varepsilon.$
	If $V=W_G$ for some simple graph $G$ on $n$ nodes and each $U^{(i)}$ is an $\mathcal{I}_n$-step function then there is a coloring $\GG$ of $G$ such that $d_\square(\UU,\WW_\GG) \leq 2k^2\varepsilon.$
\end{lemma}
\begin{proof}
	Our approach is quite elementary: consider the partition $\mathcal T$ of $[0,1]$ and $\mathcal C$ that is a pair of partitions of $[0,1]^2$ corresponding to a pair of partitions $\mathcal D$ of $[k]^2$ as above that together with $U$ describe $\UU$, and define $V^{(i)}=V \I_{C_i}$ and $\tilde V^{(i)}=(1-V) \I_{C'_i}$ for each $i \in [m]$. Then 
	\begin{align}
	\|U^{(i)}-V^{(i)}\|_\square \leq \sum_{(\alpha,\beta) \in D_i} \|(U-V)\I_{T_\alpha \times T_\beta}\|_\square \leq \varepsilon |D_i|
	\end{align}
	for each $i \in [m]$, and the analogous upper bound applies to $\|\tilde U^{(i)}-\tilde V^{(i)}\|_\square$. Summing up over $i$ gives the result stated in the lemma.
	
	The argument showing the part regarding simple graphs is identical.
	
\end{proof}
Note that in \Cref{ch6:lemma1} we required $U$ and $V$ to be close in the cut-$\PPP$-norm for some partition $\PPP$, and $U$ to be a $\PPP$ step function to guarantee for each $\UU$ the existence of $\VV$ that is close to it in the cut distance of $k$-colored digraphons.
Using the fact that in the weakly non-deterministic framework cut-closeness of instances implies the cut-closeness of the sets  of their node-$(k,m)$-colorings we can formulate the next corollary of \Cref{ch6:graphmain} that is one of the main results of this subsection.
\begin{corollary}
	Let $f$ be a weakly non-deterministically testable graph parameter with witness parameter $g$ of node-$(k,m)$-colored graphs with the corresponding sample complexity $q_g$. Then $f$ is testable with sample complexity $q_f$, and there exists a $c>0$ that does depend only on $k$ and not on $f$ so that for any $\varepsilon>0$ have $q_f(\varepsilon) \leq \exp^{(2)}(cq^2_g(\varepsilon/2))$.
\end{corollary}
\begin{proof}
	We will give only a sketch of the proof, as it is almost identical to that of \Cref{ch6:graphmain}, and we automatically refer to that, including the notation used in the current proof.
	Let $G$ be a simple graph on $n$ nodes, and let $\varepsilon>0$ be fixed, $q \geq \exp^{(2)}(cq^2_g(\varepsilon/2))$ for some constant $c>0$ that will be specified later. 
	The part concerning the lower bound of $f(\G(q,G))$ is completely identical to the general case.
	
	
	
	For the upper bound set $\Delta=\exp(-cq^2_q(\varepsilon))$. We condition on the event $\delta_\square(G, \G(q,G)) \leq \Delta$, whose failure probability is sufficiently small due to \Cref{ch3:samp}, i.e. for $q \geq 2^{100/\Delta^2}$ it is at most
	$\exp\left(-4^{100/\Delta^2} \frac{\Delta^2}{50}\right)$. We define $c$ to be large enough so that the above lower bound on $q$ holds true whenever $q\geq q_f(\varepsilon)$.  Now we select an arbitrary node-$(k,m)$-coloring $\FFF$ of $\G(q,G)$ and apply the Weak Regularity Lemma for $2m$-colored graphons, \Cref{ch3:wregk}, in the $\mathcal{I}_n$-step function case with error parameter $\Delta/(2k^2+1)$ (keeping in mind that $m\leq k^2$) to get a tuple of $\mathcal{I}_n$-step functions forming $\UU$ with at most $t'_{2k^2}(\Delta/(2k^2+1))$ steps. We define the $\mathcal{I}_n$-step function graphon $U=\sum_{i=1}^m U_i$ and note that our condition implies that $\delta_\square(G,U) \leq 2 \Delta$, since $\delta_\square(G,U) \leq \delta_\square(G,\G(q,G)) + \delta_\square(\G(q,G),U).$ 
	To finish the proof we apply \Cref{ch6:lemma3}, it implies the existence of a node-$(k,m)$-coloring $\GG$ of $G$ so that  $\delta_{\square}(\GG,\FFF) \leq (2k^2+1)\Delta$. Applying   \Cref{ch6:dev} delivers the desired result by establishing that $|g(\FFF)-g(\GG)|\leq \varepsilon.$
	
\end{proof}  

\subsection*{Second approach}
Recall the notion of layered ground state energies of $r$-arrays of \cite{BCL2} and \cite{KM1} for arbitrary $r\geq 1$.

Let $r,k \geq 1$, and $G=(G^z)_{z \in [k]^r}$ be $[k]^r$-tuple of real $r$-arrays of size $n$, and $\mathcal T=(T_1, \dots, T_k)$ a partition of $[n]$ into $k$ parts. Then 
\begin{align}
\EEE_{\mathcal T}(G)=\sum_{z \in [k]^r} \frac{1}{n^r} \sum_{i_1, \dots, i_r=1}^n G^z(i_1, \dots, i_r) \prod_{j=1}^r \I_{T_{z_j}}(i_j),
\end{align}
and 
\begin{align}
\hat \EEE(G)=\max_{\mathcal T} \EEE_{\mathcal T}(G),
\end{align}
where the maximum runs over all integer partitions $\mathcal T$ of $[n]$ into $k$ parts. 

We will make use of the next generalization of \Cref{ch4:maincor1} for hypergraphs from \cite{AVKK2} and \cite{KM1} that deals with the testability of layered ground state energies, in particular the dependence of the upper bound on the sample complexity on the dimension $r$.

\begin{theorem}\cite{AVKK2,KM1} \label{ch4:main}
	Let $r\geq 1$, $q \geq 1$, and $\varepsilon >0$. Then for any $[q]^r$-tuple of $([-\|W\|_\infty,\|W\|_\infty],r)$-graphons $W=(W^z)_{z \in [q]^r}$ and $k \geq \Theta^4  \log(\Theta) q^r$ with $\Theta=\frac{2^{r+7}q^r r}{\varepsilon}$  we have 
	\begin{align}\label{ch4:eq1}
	\PPP(|\EEE(W)-\hat \EEE(\G (k,W))|>\varepsilon \|W\|_\infty)<\varepsilon.
	\end{align}
\end{theorem}

We are ready to state and prove the other main result of the section that includes a further improvement fo the upper bound on the sample complexity compared to our first approach in the weak nondeterministic testing setting.

\begin{theorem}
	Let $r\geq 1$ and $f$ be a weakly non-deterministically testable $r$-graph parameter with witness parameter $g$ of node-$(k,m)$-colored graphs, and let the corresponding sample complexity functions be $q_f$ and $q_g$. Then $f$ is testable and there exist a $c_{r,k}>0$  that does depend only on $r$ and $k$, but not on $f$ such that for any $\varepsilon>0$ we have $q_f(\varepsilon) \leq \exp(c_{r,k} q_g(\varepsilon/8)).$ 
\end{theorem}

\begin{proof}
	Let $r\geq 1$ be arbitrary, and $f$ be a weakly nondeterministically testable $r$-graph parameter with a certificate specified by the constants $k$ and $m \leq k^r$, and the testable $2m$-colored $r$-graph parameter $g$. Then 
	\begin{align*}
	f(G)=\max_{\mathcal T, \mathcal D} g(\GG(G,\mathcal T, \mathcal D)), 
	\end{align*}
	where the maximum goes over every pair $(\mathcal T, \mathcal D)$, where  $\mathcal T$ is a partition of $V(G)$ into $k$ parts, and $\mathcal D$ is a pair of partitions of $[k]^r$ into $m$ parts, and $\GG(G,\mathcal T, \mathcal D)$\sindex[symbols]{g@$\GG(G,\mathcal T, \mathcal D)$} is the edge $2m$-colored graph defined by its parameters as seen above. 
	Define the for each fixed $\mathcal D$ the node-$k$-colored $r$-graph (i.e., a simple $r$-graph together with a $k$-coloring of its nodes) parameter $g^{\mathcal D}(G,\mathcal T)=g(\GG(G,\mathcal T, \mathcal D))$ and the  simple $r$-graph parameter $f^{\mathcal D}(G)=\max_{\mathcal T} g^{\mathcal D}(G,\mathcal T)$. 
	
	Let $\varepsilon>0$ be arbitrary, define \begin{align*}g^{\varepsilon}(\GG(G, \PPP, \mathcal D))=\sum_{F,\mathcal T} t(\GG(F,\mathcal T, \mathcal D),\GG(G,\PPP, \mathcal D))g(\GG(F,\mathcal T, \mathcal D)),\end{align*} 
	where the sum goes over all simple $r$-graphs $F$ on $q_0=q_g(\varepsilon/8)$ vertices and  partitions $\mathcal T$  of $[q_0]$ into $k$ parts.
	By the testability of $g$ we have  
	\begin{align}\label{ch6:1eq2}
	|g^{\varepsilon}(\GG(G, \PPP, \mathcal D))-g(\GG(G, \PPP, \mathcal D))|  \leq \varepsilon/4,
	\end{align}
	for each permitted tuple $(G, \PPP, \mathcal D)$. Analogously we define
	\begin{align*}g^{\varepsilon, \mathcal D}(G,\PPP)=\sum_{F,\mathcal T} t(\GG(F,\mathcal T, \mathcal D),\GG(G,\PPP, \mathcal D))g(\GG(F,\mathcal T, \mathcal D)),
	\end{align*}
	and 
	\begin{align*}f^{\varepsilon, \mathcal D}(G)=\max_{\mathcal T}g^{\varepsilon, \mathcal D}(G,\mathcal T).
	\end{align*}
	It follows from (\ref{ch6:1eq2}) that for any $G$ simple $r$-graph
	\begin{align*}
	|f^{\varepsilon, \mathcal D}(G)-f^{\mathcal D}(G)|  \leq \varepsilon/4,
	\end{align*}
	and for any $\varepsilon>0$ and $q\geq 1$ we have 
	\begin{align}\label{ch6:1eq22}
	|f(G)-f(\G(q,G))| \leq \max_{\mathcal D} |f^{\varepsilon, \mathcal D}(G)-f^{\varepsilon, \mathcal D}(\G(q,G))|+\varepsilon/2.
	\end{align}
	For any  $\varepsilon>0$ and $\mathcal D$ that is a pair of partitions of $[k]^r$ into $m$ parts the parameter $f^{\varepsilon, \mathcal D}$ can be re-written as an energy of $q_0$-arrays: For $G$ of size $[n]$ let $H=(H^z)_{z \in [k]^{q_0}}$ so that for each $z \in [k]^{q_0}$ the real $q_0$-array  $H^z$ is defined by \begin{align*}H^z(i_1, \dots, i_{q_0})= g^{\mathcal D}(G[(i_1, \dots, i_{q_0})],\PPP_z(i_1, \dots, i_{q_0}))\end{align*} for each $(i_1, \dots, i_{q_0}) \in [n]^{q_0}$, where $\PPP_z(i_1, \dots, i_{q_0})=(P_1, \dots, P_k)$ is a partition of $(i_1, \dots, i_{q_0})$ given by $P_l=\{ \, i_j \mid z_j=l\,\}$ for $l \in [k]$. Then for each $\mathcal T$ that is a partition of $[n]$ into $k$ parts we can assert that
	\begin{align*}
	g^{\varepsilon, \mathcal D}(G,\mathcal T) &= \sum_{z \in [k]^{q_0}} \frac{1}{n^{q_0}} \sum_{i_1, \dots, i_r=1}^n  g^{\mathcal D}(G[(i_1, \dots, i_{q_0})],\PPP_z(i_1, \dots, i_{q_0})) \prod_{j=1}^r \I_{T_{z_j}}(i_j) \\
	&=  \sum_{z \in [k]^{q_0}} \frac{1}{n^{q_0}} \sum_{i_1, \dots, i_r=1}^n H^z(i_1, \dots, i_{q_0}) \prod_{j=1}^r \I_{T_{z_j}}(i_j) \\
	&= \EEE_{\mathcal T} (H), 
	\end{align*}
	and further 
	\begin{align*}
	f^{\varepsilon, \mathcal D}(G)=\max_{\mathcal T} g^{\varepsilon, \mathcal D}(G,\mathcal T)=\max_{\mathcal T}  \EEE_{\mathcal T} (H)= \hat \EEE(H).
	\end{align*}
	Analogously it holds for any $q\geq q_0$ that $f^{\varepsilon, \mathcal D}(\G(q,G))= \hat \EEE(\G(q,H)).$

	This implies by \Cref{ch4:main} that for $q \geq \Theta^4  \log(\Theta)$ with $\Theta=\frac{2^{q_0+11}k^{q_0} q_0}{\varepsilon}$ and each fixed $\mathcal D$ that 
	\begin{align*}
	\PP(|f^{\varepsilon, \mathcal D}(G)-f^{\varepsilon, \mathcal D}(\G(q,G))| > \varepsilon/2) <2 \exp\left( -\frac{\varepsilon^2q}{32 q_0^2} \right).
	\end{align*}
	The probability that the event in the previous formula occurs for some $\mathcal D$ is at most $k^{2r k^r} 2 \exp\left( -\frac{\varepsilon^2q}{32 q_0^2} \right)$, therefore by recalling (\ref{ch6:1eq22}) we can conclude that there exists a constant $c_{r,k}>0$ not depending on other specifics of $f$ such that for each simple graph $G$ and $q \geq \exp(c_{r,k}q_g(\varepsilon/8))$ it holds that 
	\begin{align*}
	\PP(|f(G)-f(\G(q,G))| > \varepsilon) < \varepsilon.
	\end{align*}
	
\end{proof}

\section{Further research}

The sample complexity upper bounds provided in this paper in both the general and the special case are not known to be sharp, moreover, the lower bounds available at the moment are only trivial ones. An interesting open question is to improve both upper and lower bounds in the above setting, perhaps upper bound conditions (such as polynomial testability) for the witness parameter are of relevance here.


\section*{Acknowledgement}

We thank Laci Lov\'asz for an interesting discussion connected to the subject of this paper.

\bibliographystyle{notplainnat}
\bibliography{thesis_refs}

\end{document}